 \documentclass[11pt]{article}
\usepackage[margin=1in]{geometry}

\usepackage{amsthm}   % comment out following 2 for IEEE,CDC papers
\usepackage{amssymb}
\usepackage{mathtools}
\usepackage{enumitem}
\usepackage{amsfonts}
\usepackage{graphicx}
\usepackage{textcomp}
\usepackage{nicefrac}
\usepackage{algorithm}
\usepackage{algorithmic}
\usepackage{siunitx}
\usepackage{setspace}
\usepackage{url}
\usepackage{graphicx}
\usepackage{subcaption}
\usepackage{cite}
\usepackage{placeins}

\usepackage{tikz}
    \usetikzlibrary{positioning,calc}
    \usetikzlibrary{external}
\usepackage{pgfplots}
    \pgfplotsset{compat=1.18}
    
\allowdisplaybreaks

\newtheorem{lemma}{Lemma}
\newtheorem{theorem}{Theorem}

\newtheorem{definition}{Definition}

\usepackage[
colorlinks=true,
linkcolor=black,
citecolor=black,
urlcolor=blue
]{hyperref}

\usepackage[nameinlink]{cleveref}

\usepackage{xcolor}

\usepackage{cleveref}
\crefname{section}{Sec.}{Secs.}
\Crefname{section}{Sec.}{Secs.}

\crefname{figure}{Fig.}{Figs.}
\Crefname{figure}{Fig.}{Figs.}

\crefname{proposition}{Prop.}{Props.}
\Crefname{proposition}{Prop.}{Props.}

\crefname{theorem}{Thm.}{Thms.}
\Crefname{theorem}{Thm.}{Thms.}

\crefname{lemma}{Lem.}{Lems.}
\Crefname{lemma}{Lem.}{Lems.}

\crefname{corollary}{Cor.}{Cors.}
\Crefname{corollary}{Cor.}{Cors.}

\crefname{remark}{Remark}{Remarks}
\Crefname{remark}{Remark}{Remarks}

\crefname{assumption}{Assumption}{Assumptions}
\Crefname{assumption}{Assumption}{Assumptions}

\crefname{definition}{Def.}{Defs.}
\Crefname{definition}{Def.}{Defs.}
\crefname{example}{Example}{Examples}
\Crefname{example}{Example}{Examples}

\crefname{table}{Table}{Tables}
\Crefname{table}{Table}{Tables}

\title{\LARGE \bf
Switching Observers for Linear Systems: Beyond Individual
Observability
}

\author{Masoud S. Sakha and Rushikesh Kamalapurkar%
\thanks{Masoud S. Sakha is a Ph.D. candidate, and Rushikesh Kamalapurkar is with the Department of Mechanical and Aerospace Engineering, University of Florida, Gainesville, FL 32608 USA. (Emails: {\tt\small masoud.sakha@ufl.edu, rkamalapurkar@ufl.edu}).}%
}

\begin{document}

\maketitle
\thispagestyle{empty}
\pagestyle{empty}

%%%%%%%%%%%%%%%%%%%%%%%%%%%%%%%%%%%%%%%%%%%%%%%%%%%%%%%%%%%%%%%%%%%%%%%%%%%%%%%%
%%%%%%%%%%%%%%%%%%%%%%%%%%%%%%

\begin{abstract}
This paper develops an LMI-based switching observer for linear time-invariant systems with two output channels, where only one output is available at any given time. We assume that the system is observable when both outputs are considered together, and that the time intervals between consecutive output switches are uniformly bounded. Under these assumptions, we design a set of channel-dependent Luenberger observer gains offline via linear matrix inequalities and show that the resulting switched observer drives the estimation error to zero. The main contribution is an offline gain design and convergence result for switched observers that does not require individual output channels to be detectable.
\end{abstract}

% \begin{IEEEkeywords}
% Switching observer, Linear time-invariant (LTI) systems, Observability, Intermittent outputs, Luenberger observer, Linear matrix inequalities (LMI)
% \end{IEEEkeywords}

\section{Introduction}\label{Sec_Intro}

In many engineering systems, it is not always possible or advisable to access every available output in each control or sensing loop. Constraints in sensing, communication, or coverage may cause the set of available outputs to vary over time. For example, in aerospace applications, different sensors may operate only in specific flight regimes, while in robotic systems, line-of-sight limitations can render certain measurements unavailable in parts of the workspace. These situations naturally motivate observer architectures in which the correction term is updated using only the currently available measurement channel.

A natural first step in addressing such limitations is to determine which sensors should be used. Sensor placement is closely tied to the concept of observability, which characterizes the ability to reconstruct the system state from output measurements \cite{SCC.Willems.Mitter1971}. Observability-based metrics have therefore been widely used to guide sensor placement in dynamical systems. In this setting, the goal is to select a fixed subset of measurement channels that provides sufficient or maximal information about the state, often using criteria derived from observability Gramians or related system-theoretic measures. Various computational approaches have been developed for large-scale systems, including optimization-based and evolutionary methods \cite{SCC.Sakha.Shaker2017,SCC.SeyedSakha.Shaker2017}. For switched systems, related Gramian-based criteria, including \textit{nice observability Gramians}, have also been used for sensor placement \cite{SCC.Petreczky.Wisniewski.ea2013,SCC.SeyedSakha.Shaker.ea2019}. However, these results focus on static selection of sensors and assume that all selected measurements are simultaneously available.

When simultaneous access to all sensors is not possible, one must instead determine how to utilize the available measurements over time. Sensor scheduling considers the dynamic selection of sensors or outputs over time, in contrast to sensor placement which is typically performed offline. In many practical systems, limitations in sensing, communication bandwidth, or energy prevent simultaneous access to all measurements, necessitating online selection of a subset of sensors. Early work such as \cite{SCC.Meier.Peschon.ea1967} formulated this problem as the optimal control of measurement subsystems, while subsequent developments introduced stochastic and optimal scheduling strategies \cite{SCC.Gupta.Chung.ea2006,SCC.Shi.Cheng.ea2011}. More recent research has focused on computationally efficient scheduling policies for linear dynamical systems, including convex optimization and submodular approaches \cite{SCC.Vitus.Zhang.ea2012,SCC.Jawaid.Smith2015}, as well as covariance-based formulations that directly optimize estimation performance \cite{SCC.Maity.Hartman.ea2022}. These formulations naturally lead to time-varying measurement processes in which the output map evolves according to a scheduling signal, thereby connecting sensor scheduling with switched and time-varying estimation frameworks.

Given such time-varying measurement availability, the problem of reconstructing the system state becomes central. State estimation aims to reconstruct the system state from input–output measurements. For linear systems, the Luenberger observer \cite{SCC.Luenberger1964} and Kalman filter \cite{SCC.Kalman1960a,SCC.Kalman.Bucy1961} provide systematic frameworks for state estimation. These approaches rely on structural properties such as observability or detectability to ensure convergence of the estimation error.

In many practical systems, the set of available measurements may be determined by a designed scheduling policy or dictated by external conditions. In the former case, the selection is governed by sensing constraints, while in the latter, output availability depends on factors such as location, visibility, or communication status. In both scenarios, the observer must ensure stability of the estimation error despite temporary loss of information from some sensors. A common approach is to treat the output map as time-varying and applies observer design tools for linear time-varying systems, where stability is established under uniform observability or persistence-of-excitation conditions over sliding time windows \cite{SCC.Yin.Gu.ea2022,SCC.Xu.Xu.ea2022}. In such approaches, the observer gain is typically time-varying and must be generated online (e.g., via Riccati-based implementations), which increases implementation complexity and relies on strong uniform informativeness assumptions on the measured outputs.

Switched observers provide a natural framework for state estimation in systems with switching dynamics and time-varying measurement structures. Early results on switched observers focused on mode-dependent observer design using linear matrix inequality (LMI)-based synthesis and dwell-time conditions to guarantee stability of the estimation error \cite{SCC.Chen.Mehrdad2004,SCC.Pettersson2006}. These approaches were further extended to robust and unknown-input observer designs, where disturbances and model uncertainties are explicitly taken into account \cite{SCC.Bejarano.Pisano2010}.

Beyond stability analysis, several works have investigated observability properties of switched systems, showing that observability may be recovered through switching even when individual subsystems are not observable \cite{SCC.Tanwani.Shim.ea2012}. The possibility of recovering observability through switching has led to observer architectures that exploit information across mode transitions, including switch-based observers capable of estimating both the state and the switching signal \cite{SCC.Kuesters.Trenn.ea2017}. Extensions to nonlinear and descriptor systems include observer design for switched neural networks \cite{SCC.Lian.Feng.ea2011}, functional and reduced-order observers for switched and descriptor systems \cite{SCC.Huang.Che.ea2021}, and time-scheduled observer designs for systems with unknown inputs \cite{SCC.Shi.Fei.ea2022}. More recently, switched observer frameworks have been developed for systems that are locally unobservable or exhibit time-varying observability properties, where switching between observer gains associated with observable regions can ensure convergence under suitable dwell-time conditions \cite{SCC.Aranovskiy.Efimov.ea2021}. Similar results are available for nonlinear systems using Lipschitz continuity and linear parameter-varying (LPV)-based formulations with LMI synthesis \cite{SCC.YupanquiTello.Coutinho.ea2024}.

Despite these advances, most existing switched-observer designs rely on the assumption that each subsystem associated with a single output channel is individually detectable or observable. Under this assumption, stability of the estimation error is typically established using common Lyapunov functions, multiple Lyapunov functions, or dwell-time arguments. An alternative line of work considers time-varying or online observer designs, where the observer gain is updated in real time based on the available measurements. These approaches can relax the requirement of subsystem-wise observability by exploiting uniform or collective observability over time; however, they typically rely on strong uniform informativeness conditions and require online gain computation, which increases implementation complexity.

In this paper, we consider a linear time-invariant system with two output channels under a time-based switching mechanism in which only one measurement is available at a time. We consider switching signals satisfying uniform bounds on the residence times of each channel and establish convergence of the estimation error under such switching signals. Although the proposed framework naturally extends to systems with multiple output channels with minor modifications, we focus on the case of two outputs to simplify the exposition.

To formalize the switching mechanism, we assume that the observer sequentially accesses the two output channels, with the time spent on each channel uniformly bounded above and below. The switching pattern is described in terms of cycles indexed by $k$, parameterized by $(\tau_k,s_k)$, where $\tau_k$ denotes the cycle length and $s_k \in (0,1)$ specifies the fraction of time assigned to the first channel.

Under this structure, we develop LMI-based conditions to design a pair of Luenberger observer gains $(L_{1k},L_{2k})$ for each cycle. The gains are piecewise constant within each cycle and are computed offline as functions of the cycle parameters $(\tau_k,s_k)$ \cite{SCC.Boyd.ElGhaoui.ea1994}. Stability is established by analyzing the evolution of the estimation error across successive switching cycles.

The proposed framework extends classical switched-observer designs in two key directions: first, it does not require individual output channels to be detectable or observable, and second, it enables offline synthesis of observer gains,  and does not rely on online, time-varying gain computation. As a result, the approach applies to systems with complementary outputs and intermittent measurement access, while maintaining a computationally tractable  observer structure.

\section{Problem Statement}\label{Sec_ProbFormul}

A linear time-invariant (LTI) system with scheduled output is described by
\begin{equation}\label{e:system}
\begin{aligned}
\dot{x}(t) &= Ax(t)+Bu(t),\\
y(t) &= C_{\rho(t)}x(t), \qquad 
\rho(.):\mathbb{R}_{\ge 0}\to\{1,2\},
\end{aligned}
\end{equation}
where $x(t)\in\mathbb{R}^n$ is the state, 
$u(t)\in\mathbb{R}^m$ is the control input, and 
$y(t)\in\mathbb{R}^{p_{\rho(t)}}$ is the available measurement. 
The matrices $A\in\mathbb{R}^{n\times n}$ and 
$B\in\mathbb{R}^{n\times m}$ are known.

The plant is equipped with two output channels
\begin{equation}\label{e:outputs}
y_1(t)=C_1x(t), 
\qquad 
y_2(t)=C_2x(t),
\end{equation}
where $C_1\in\mathbb{R}^{p_1\times n}$ and 
$C_2\in\mathbb{R}^{p_2\times n}$. 
Due to sensing or communication constraints, only one output channel is available at any time, 
as encoded by $C_{\rho(t)}\in\{C_1,C_2\}$ in \eqref{e:system}.

The individual pairs $(A,C_i)$, $i\in\{1,2\}$, are not assumed to be either observable or  detectable. 
However, when both outputs are accessible simultaneously, the aggregated measurement
$y(t)=Cx(t)$, with
\begin{equation}\label{e:C_agg}
C \coloneqq 
\begin{bmatrix}
C_1\\
C_2
\end{bmatrix},
\end{equation}
renders the combined pair $(A,C)$ observable.

The output availability is determined by the switching signal $\rho(\cdot)$ in \eqref{e:system}. 
Let $\{t_\ell\}_{\ell\in\mathbb{Z}_{\ge 0}}$ denote the switching instants satisfying
\begin{equation}\label{e:switching_instants}
0=t_0<t_1<t_2<\cdots .
\end{equation}

\medskip
Throughout the paper, we impose the following assumptions:

\smallskip
\noindent
\textbf{(A1) observability:}
 $(A,C)$ is observable.

\smallskip
\noindent
\textbf{(A2) Uniform boundedness of channel residence times:}
There exists a constant $\underline{h}>0$ and $ \bar h>0$ such that the switching intervals satisfy
\begin{equation}\label{e:switching_interval_ub}
\underline{h}
\;<\;
t_{\ell+1}-t_\ell
\;<\;
\bar h,
\qquad \forall \ell\in\mathbb{Z}_{\ge 0}.
\end{equation}

Since there are only two output channels, once the initial channel is specified, the entire alternating schedule is uniquely determined. Moreover, the labeling of the channels is arbitrary; therefore, without loss of generality, we index the switching signal so that
\[
\rho(t)=
\begin{cases}
1, & t \in [t_{2k},\, t_{2k+1}),\\
2, & t \in [t_{2k+1},\, t_{2k+2}),
\end{cases}
\quad k \in \mathbb{Z}_{\ge 0}.
\]

Two consecutive switching intervals are grouped into one switching cycle by setting 
$\ell=2k$, $k\in\mathbb{Z}_{\ge 0}$. 
The  length and the duty fraction  of the $k$th   switching cycle are defined by
\begin{equation}\label{e:tau_k}
\tau_k \coloneqq t_{2k+2}-t_{2k},
\end{equation}
\begin{equation}\label{e:s_k}
s_k \coloneqq
\frac{t_{2k+1}-t_{2k}}{\tau_k}
\in(0,1).
\end{equation}
Thus, during cycle $k$, output $y_1$ is available for $s_k\tau_k$ seconds on $[t_{2k},t_{2k+1})$, 
and output $y_2$ is available for $(1-s_k)\tau_k$  seconds on $[t_{2k+1},t_{2k+2})$. 
The pair $(\tau_k,s_k)$ may vary from cycle to cycle.

From \eqref{e:switching_interval_ub}, the cycle length satisfies

\begin{equation}\label{e:cycle_length_ub}
 \underline{\tau} 
\le
\tau_k
\le
\bar{\tau},
\qquad \forall k\in\mathbb{Z}_{\ge 0},
\end{equation}

where one may take $\underline{\tau}\coloneqq 2\underline {h}$ and $\bar{\tau}\coloneqq 2\bar h$.

The state-estimation dynamics are given by a switching Luenberger observer of the form
\begin{equation}\label{e:switched_observer}
\dot{\hat{x}}(t)
=
A\hat{x}(t)+Bu(t)
+
L(t)\big(y(t)-C_{\rho(t)}\hat{x}(t)\big),
\end{equation}
where $\hat{x}(t)\in\mathbb{R}^n$ denotes the state estimate.

The observer gain is piecewise constant on each switching interval; specifically, there exists a 
sequence $\{L_\ell\}_{\ell\in\mathbb{Z}_{\ge 0}}$ such that
\begin{equation}\label{e:L_of_t_def}
L(t)\coloneqq L_\ell,
\qquad 
t\in[t_\ell,t_{\ell+1}),
\quad 
\ell\in\mathbb{Z}_{\ge 0}.
\end{equation}
Equivalently, over cycle $k$ the gain satisfies
\begin{equation}\label{e:L_sigma_cycle}
L(t)=
\begin{cases}
L_{1,k}, & t\in[t_{2k},\,t_{2k+1}),\\[4pt]
L_{2,k}, & t\in[t_{2k+1},\,t_{2k+2}),
\end{cases}
\qquad k\in\mathbb{Z}_{\ge 0}.
\end{equation}

Defining the estimation error as
\begin{equation}\label{e:est_error}
\tilde{x}(t)\coloneqq  \hat{x}(t)-x(t),
\end{equation}

the error dynamics satisfy
\begin{equation}\label{e:error_dynamics}
\dot{\tilde{x}}(t)
=
\big(A-L(t)C_{\rho(t)}\big)\tilde{x}(t).
\end{equation}
The objective is to design the gain sequence 
$\{L_{1,k},L_{2,k}\}_{k\in\mathbb{Z}_{\ge 0}}$ 
such that $\|\tilde{x}(t)\|\to 0$ as $t\to\infty$.

\section{Preliminaries}\label{Sec_Prelim}

This section recalls standard tools used in the analysis of the switched
estimation-error dynamics. 
Unless stated otherwise, $\|\cdot\|$ denotes the induced $2$-norm, and
$\mu(\cdot)$ its associated logarithmic norm. Other induced norms are
denoted by $\|\cdot\|_p$ and $\mu_p(\cdot)$.

%%%%%%%%%%%
% \begin{lemma}[Triangle inequality]\label{t:triangle}
% Let $\|\cdot\|_p$ be an induced matrix norm. Then, for any matrices $A$ and $B$ of compatible dimensions,
% \begin{equation}\label{e:triangle}
% \|A+B\|_p \le \|A\|_p + \|B\|_p.
% \end{equation}
% \end{lemma}

% \begin{lemma}[Submultiplicativity of induced norms]\label{t:submult}
% Let $\|\cdot\|_p$ be an induced matrix norm. Then, for any matrices $A$ and $B$ of compatible dimensions,
% \begin{equation}\label{e:submult}
% \|AB\|_p \le \|A\|_p\,\|B\|_p.
% \end{equation}
% \end{lemma}

% \begin{lemma}[Singular-value bounds]\label{t:sv_bounds}
% Let $H\in\mathbb R^{m\times n}$ and let $\rho_{\min}(H)$ and $\rho_{\max}(H)$ denote its smallest and largest
% singular values. Then, for all $x\in\mathbb R^{n}$,
% \begin{equation}\label{e:sv_bounds_two_sided}
% \rho_{\min}(H)\,\|x\|\ \le\ \|Hx\|\ \le\ \rho_{\max}(H)\,\|x\|.
% \end{equation}

% \end{lemma}
%%%%%%%%%%%%%%%%%%%%%%%%%%%%%%%%%%%%%%%%%%%%%%%%%%%%%%%%%%%%%%%%%%
The following standard inequalities from linear algebra will be used throughout the analysis.

\smallskip
\begin{lemma}[Standard norm inequalities]\label{t:basic_norm}
Let $\|\cdot\|_p$ be an induced matrix norm. Then, the following hold:
\begin{itemize}
\item For any matrices $A$ and $B$ of compatible dimensions:
\begin{itemize}
\item \textbf{(Triangle inequality)}
\begin{equation}\label{e:triangle}
\|A+B\|_p \le \|A\|_p + \|B\|_p.
\end{equation}
\item \textbf{(Submultiplicativity of induced norms)}
\begin{equation}\label{e:submult}
\|AB\|_p \le \|A\|_p\,\|B\|_p.
\end{equation}
\end{itemize}
\item \textbf{(Singular-value bounds)} For any matrix $H\in\mathbb R^{m\times n}$, with smallest and largest singular values
$\rho_{\min}(H)$ and $\rho_{\max}(H)$, respectively, and any vector $x\in\mathbb R^{n}$,
\begin{equation}\label{e:sv_bounds_two_sided}
\rho_{\min}(H)\,\|x\|
\le
\|Hx\|
\le
\rho_{\max}(H)\,\|x\|.
\end{equation}
\end{itemize}
\end{lemma}

%%%%%%%%%%%%%%%%%%%%%%%%%%%%%%%%%%%%%%%%%%%%%%%%%%%%%%%%%%%%%

% \begin{definition}[Logarithmic norm (matrix measure) \cite{SCC.Soederlind2006}]\label{d:log_norm}
% Let $\|\cdot\|_p$ be an induced matrix norm. The corresponding logarithmic norm of $A\in\mathbb{R}^{n\times n}$
% is defined by
% \begin{equation}\label{e:mu_p_def}
% \mu_p(A)\coloneqq \lim_{h\to 0^+}\frac{\|I+hA\|_p-1}{h}.
% \end{equation}
% For $p=2$, one has
% \begin{equation}\label{e:mu2_formula}
% \mu(A)=\lambda_{\max}\!\left(\frac{A+A^\top}{2}\right),
% \end{equation}
% where $\lambda_{\max}(\cdot)$ denotes the largest eigenvalue.
% \end{definition}

% \begin{lemma}[Matrix exponential bound\cite{SCC.Soederlind2006}]\label{t:log_norm_bound}
% Let $\|\cdot\|_p$ be an induced matrix norm and let $\mu_p(A)$ be defined in \eqref{e:mu_p_def}. Then, for all $t\ge 0$,
% \begin{equation}\label{e:exp_bound_mu_p}
% \|e^{At}\|_p \le e^{t\,\mu_p(A)}.
% \end{equation}
% \end{lemma}

\smallskip
\begin{definition}[Logarithmic norm {\cite{SCC.Soederlind2006}}]\label{d:log_norm}
Let $\|\cdot\|_p$ be an induced matrix norm. The corresponding logarithmic norm of $A\in\mathbb{R}^{n\times n}$
is defined by
\begin{equation}\label{e:mu_p_def}
\mu_p(A)\coloneqq \lim_{h\to 0^+}\frac{\|I+hA\|_p-1}{h}.
\end{equation}
For $p=2$, one has
\begin{equation}\label{e:mu2_formula}
\mu(A)=\lambda_{\max}\!\left(\frac{A+A^\top}{2}\right),
\end{equation}
where $\lambda_{\max}(\cdot)$ denotes the largest eigenvalue.
5
Moreover, the matrix exponential satisfies, for all $t\ge 0$,
\begin{equation}\label{e:exp_bound_mu_p}
\|e^{At}\|_p \le e^{t\,\mu_p(A)}.
\end{equation}
\end{definition}

%%%%%%%%%%%%%%%%%%%%%%%%%%%%%%%%%%%%%%%%%%%%%%%%%%%%%%%%%%%%%%%
\smallskip
\begin{definition}[Kernel (null space)]\label{def:kernel}
For a matrix $S\in\mathbb R^{p\times q}$, the kernel (null space) of $S$ is defined as
\begin{equation*}
\ker(S) := \{x\in\mathbb R^{q} \mid Sx = 0\}.
\end{equation*}
\end{definition}

%%%%%%%%%%%%%%%%%%%%%%%%%%%%%%%%%%%%%

Throughout the remainder of this section, observability-related notions are stated for the standard LTI plant
\begin{equation}\label{e:lti_single_output_prelim}
\begin{aligned}
\dot{x}(t) &= Ax(t)+Bu(t),\\
y(t) &= Cx(t),
\end{aligned}
\end{equation}
where $x(t)\in\mathbb{R}^n$, $u(t)\in\mathbb{R}^m$, and $y(t)\in\mathbb{R}^p$, and where $A$, $B$, and $C$ are constant matrices.

\smallskip
\begin{lemma}[Variation-of-constants formula]\label{t:var_constants}
Consider the system \eqref{e:system}. For any measurable input 
$u:[t_0,t_1]\to\mathbb{R}^m$, the state satisfies, for all $t\in[t_0,t_1]$,
\begin{equation}\label{e:var_constants}
x(t)=e^{A(t-t_0)}x(t_0)+\int_{t_0}^{t} e^{A(t-\tau)}Bu(\tau)\,d\tau,
\end{equation}
(see Sec.~6.1 of \cite{SCC.Hespanha2018}).
\end{lemma}

\begin{definition}[Observability and detectability]\label{d:obs_detect}
The observability matrix of the pair $(A,C)$ is defined as
\begin{equation}\label{e:obs_matrix}
\mathcal{O}(A,C) \coloneqq
\begin{bmatrix}
C\\
CA\\
\vdots\\
CA^{n-1}
\end{bmatrix}.
\end{equation}

The pair $(A,C)$ is said to be \emph{observable} if
\begin{equation}\label{e:observability_rank}
\operatorname{rank}\!\big(\mathcal{O}(A,C)\big)=n.
\end{equation}

Let $\mathcal{N}(A,C)\coloneqq \ker\!\big(\mathcal{O}(A,C)\big)$ denote the unobservable subspace.
The pair $(A,C)$ is said to be \emph{detectable} if for all $x_0\in\mathcal{N}(A,C)$,
\begin{equation}\label{e:detectability_condition}
\lim_{t\to\infty} e^{At}x_0 = 0.
\end{equation}
\end{definition}

%%%%%%%%%%%%%%%%%%%%%%%%%%%%%%%%%%%%%%%%%%%%%%%%%%%%%%%%%%%
\smallskip
\begin{lemma}[Stabilizing observer gain]
\label{t:stab_observer_gain}
If the pair $(A,C)$ is detectable, then there exists
$L\in\mathbb{R}^{n\times p}$ such that $\Lambda \coloneqq A-LC$ is Hurwitz 
(see Sec.~16.6 of \cite{SCC.Hespanha2018}).
\end{lemma}

\smallskip
\begin{lemma}[Lyapunov equation]
\label{t:lyap_Lambda}
Let $\Lambda\in\mathbb{R}^{n\times n}$ be Hurwitz. Then, for any
$Q=Q^\top\succ0$, there exists a unique $P=P^\top\succ0$ such that
\begin{equation}\label{e:lyap_eq_Lambda}
\Lambda^\top P + P\Lambda = -Q,
\end{equation}
(see Sec.~8.5 of \cite{SCC.Hespanha2018}).
\end{lemma}

% \input{Sections/Informative.tex}

%%%%%%%%%%%%%%%%%%%%%%%%%%%%%%%%%%%%%%%%%%%%%%%%%%%%%%%%%%%%%%%%%%%%%%%%%%%%%%%%%%%%%%%%%%%%%%%%%%%%%%%%%%%%%%%%%%%%%%%%

\section{Switching Observer Dynamics}\label{s:Switch_Obs}

To analyze the estimation-error dynamics under switching, we exploit a similarity transformation associated with each output channel. This transformation decomposes the state into observable and unobservable components with respect to the active output. As a result, the observable part can be directly influenced by a Luenberger correction term, while the unobservable part evolves autonomously. This structure enables the design of channel-dependent observer gains that stabilize the observable components, and, through the switching mechanism, ensures convergence of the full-state estimation error.

\smallskip
\begin{lemma}[Orthogonal observability decomposition]\label{t:decomposition}
Let \((A,C)\) be a linear pair. Then there exists an orthogonal matrix
\begin{equation}\label{e:T_orthogonal}
T=\begin{bmatrix}T_o & T_u\end{bmatrix},
\qquad
T^\top T=I,
\end{equation}
such that
\begin{equation}\label{e:A_obs_decomp}
T^\top A T
=
\begin{bmatrix}
A_o & 0\\
A_{uo} & A_u
\end{bmatrix},
\end{equation}
and
\begin{equation}\label{e:C_obs_decomp}
CT=\begin{bmatrix}C_o & 0\end{bmatrix},
\end{equation}
where \((A_o,C_o)\) is observable.
\end{lemma}

\begin{proof}
See the Appendix.
\end{proof}

Applying Lemma~\ref{t:decomposition} to each output channel \(i\in\{1,2\}\), there exists an orthogonal transformation \(T_i\) such that
\begin{equation}\label{e:x_transform}
x(t) = T_i x^{(i)}(t),
\end{equation}
where the transformed state \(x^{(i)}\) can be partitioned as
\[
x^{(i)}(t) =
\begin{bmatrix}
x_o^{(i)}(t)\\
x_u^{(i)}(t)
\end{bmatrix},
\]
with the observable component \(x_o^{(i)} \in \mathbb{R}^{r_i}\), where
\[
r_i = \operatorname{rank}\!\big(\mathcal{O}(A,C_i)\big),
\]
and the unobservable component \(x_u^{(i)} \in \mathbb{R}^{n-r_i}\).

That is, there exist matrices \(A^{(i)}=T_i^\top A T_i\), \(B^{(i)}=T_i^\top B\), and \(C^{(i)}=C_i T_i\) of the form
\begin{equation}\label{e:block_structure}
A^{(i)}=
\begin{bmatrix}
A_o^{(i)} & 0\\
A_{uo}^{(i)} & A_u^{(i)}
\end{bmatrix},
\quad
B^{(i)}=
\begin{bmatrix}
B_o^{(i)}\\
B_u^{(i)}
\end{bmatrix},
\quad
C^{(i)}=\begin{bmatrix} C_o^{(i)} & 0 \end{bmatrix},
\end{equation}
such that
\begin{equation}\label{e:xi_dyn}
\begin{aligned}
\dot{x}^{(i)}(t) &= A^{(i)}x^{(i)}(t)+B^{(i)}u(t), \\
y_i(t) &= C^{(i)}x^{(i)}(t),
\end{aligned}
\end{equation}
and the pair \(\big(A_o^{(i)},C_o^{(i)}\big)\) is observable.

%%%%%%%%%%%%%%%%%%%%%%%%%%%%%%%%%%%%%%%%%%%%%%%%%%%%%%%%%%%%
\subsection{Channel-Dependent Observer}

For each channel $i\in\{1,2\}$, we design a Luenberger observer that corrects
only the observable component. To avoid excessive indexing, the possible cycle index $k$
is suppressed in this subsection and reintroduced in
Section~\ref{Sec_SwitchObs:SwitchedImplementation}.

Define the estimate in the transformed coordinate as
\begin{equation}\label{e:xhat_decomp}
\hat{x}^{(i)}(t)
=
\begin{bmatrix}
\hat{x}_o^{(i)}(t)\\
\hat{x}_u^{(i)}(t)
\end{bmatrix},
\qquad
\hat{x}(t)=T_i\hat{x}^{(i)}(t).
\end{equation}
The observer dynamics in transformed coordinates are
%%%%%%%%%%%%%%%%%%%%%%%%%
\begin{align}
\dot{\hat{x}}_o^{(i)}(t)
&=
A_o^{(i)}\hat{x}_o^{(i)}(t)
+
B_o^{(i)}u(t)
\\
&\quad
+
L_o^{(i)}
\Big(
y_i(t)-C_o^{(i)}\hat{x}_o^{(i)}(t)
\Big),
\label{e:xhat_o_dyn}
\\
\dot{\hat{x}}_u^{(i)}(t)
&=
A_{uo}^{(i)}\hat{x}_o^{(i)}(t)
+
A_u^{(i)}\hat{x}_u^{(i)}(t)
+
B_u^{(i)}u(t).
\label{e:xhat_u_dyn}
\end{align}
%%%%%%%%%%%%%%%%%%%%%%%%%%%%%%
\subsection{Estimation Error Dynamics}

Let $\tilde{x}^{(i)}(t)=\hat{x}^{(i)}(t)-x^{(i)}(t)$ and partition 
it using the transformation matrix $T_i$ to get
\begin{equation}\label{e:xtilde_decomp}
\tilde{x}^{(i)}(t)
=T_i^{-1} \tilde{x}(t)= 
\begin{bmatrix}
\tilde{x}_o^{(i)}(t)\\
\tilde{x}_u^{(i)}(t)
\end{bmatrix}.
\end{equation}
From \eqref{e:xhat_o_dyn} and \eqref{e:xhat_u_dyn} we conclude that the estimation error dynamics satisfy
%%%
\begin{align}
\dot{\tilde{x}}_o^{(i)}(t)
&=
\Lambda_o^{(i)}\,\tilde{x}_o^{(i)}(t),
\label{e:xtilde_o_dyn}\\
\dot{\tilde{x}}_u^{(i)}(t)
&=
A_{uo}^{(i)}\,\tilde{x}_o^{(i)}(t)
+
A_u^{(i)}\,\tilde{x}_u^{(i)}(t).
\label{e:xtilde_u_dyn}
\end{align}
%%%
where $\Lambda_o^{(i)} \coloneqq A_o^{(i)}-L_o^{(i)}C_o^{(i)}$.
Since $\big(A_o^{(i)},C_o^{(i)}\big)$ is observable, $L_o^{(i)}$ can be selected
such that $\Lambda_o^{(i)}$ is Hurwitz, yielding exponential decay of
$\tilde{x}_o^{(i)}(t)$.

%%%%%%%%%%%%%%%%%%%%%%%%%%%%%%%%%%%%%%%%%%%%%%%%%%%%%%%%
\subsection{Gain Embedding in Original Coordinates}

The observer gains are designed in the transformed coordinates, where the system admits an observable/unobservable decomposition with respect to each output channel. However, the observer dynamics are implemented in the original state coordinates. Therefore, this subsection derives the corresponding gains in the original coordinates by mapping them through the inverse transformation, ensuring that the designed correction terms can be directly applied to the original system.

Let $0_{u}^{(i)} \in \mathbb{R}^{(n-r_i)\times p_i}$ denote a zero matrix, and define the gain in transformed coordinates as
\begin{equation}\label{e:L_transformed}
L^{(i)}
=
\begin{bmatrix}
L_o^{(i)}\\
0_{u}^{(i)}
\end{bmatrix}.
\end{equation}
The corresponding gain in the original coordinates is given by
\begin{equation}\label{e:L_original}
L_i = T_i L^{(i)}.
\end{equation}

%%%%%%%%%%%%%%%%%%%%%%%%%%%%%%%%%%%%%%%%%%%%%%%%%%%%%%%%%%
\subsection{Switched Implementation}\label{Sec_SwitchObs:SwitchedImplementation}

We now reintroduce the cycle index $k$ to describe the switched implementation.
Over the $k$th cycle $[t_{2k},t_{2k+2})$, the observer uses the channel-dependent
gains $L_{1,k}$ and $L_{2,k}$ on the two subintervals. Defining the full-state
estimation error as $\tilde{x}(t)=\hat{x}(t)-x(t)$, the error dynamics satisfy
\begin{equation}\label{e:xtilde_dyn_piecewise}
\dot{\tilde{x}}(t)=
\begin{cases}
\big(A-L_{1,k}C_1\big)\tilde{x}(t),
& t \in [t_{2k},\, t_{2k+1}), \\[2pt]
\big(A-L_{2,k}C_2\big)\tilde{x}(t),
& t \in [t_{2k+1},\, t_{2k+2}),
\end{cases}
\end{equation}
for all $k\in\mathbb{Z}_{\ge 0}$.

\noindent
Note that the matrices $A-L_{i,k}C_i$ are not required to be Hurwitz.
Indeed, since $(A,C_i)$ may be undetectable, the estimation
error may temporarily grow on a subinterval. Stability is instead established
through a cycle-wise analysis that exploits information accumulation across the
two outputs over each switching cycle.

%%%%%%%%%%%%%%%%%%%%%%%%%%%%%%%%%%%%%%%%
The observable estimation-error components satisfy the following switched dynamics over the $k$th cycle $[t_{2k},t_{2k+2})$:
\begin{equation}\label{e:xo1_dynamics_piecewise}
\dot{\tilde x}_o^{(1)}(t)=
\begin{cases}
\Lambda_{o,k}^{(1)}\,\tilde x_o^{(1)}(t),
& t \in [t_{2k},\, t_{2k+1}), \\[4pt]
A_o^{(1)}\,\tilde x_o^{(1)}(t),
& t \in [t_{2k+1},\, t_{2k+2}),
\end{cases}
\end{equation}
and
\begin{equation}\label{e:xo2_dynamics_piecewise}
\dot{\tilde x}_o^{(2)}(t)=
\begin{cases}
A_o^{(2)}\,\tilde x_o^{(2)}(t),
& t \in [t_{2k},\, t_{2k+1}), \\[4pt]
\Lambda_{o,k}^{(2)}\,\tilde x_o^{(2)}(t),
& t \in [t_{2k+1},\, t_{2k+2}).
\end{cases}
\end{equation}

The corresponding closed-form solutions are
\begin{equation}\label{e:xo1_time_development}
\tilde x_o^{(1)}(t)=
\begin{cases}
e^{\Lambda_{o,k}^{(1)}(t-t_{2k})}\,\tilde x_o^{(1)}(t_{2k}),
& t\in[t_{2k},t_{2k+1}),\\[2mm]
e^{A_o^{(1)}(t-t_{2k+1})}\,\tilde x_o^{(1)}(t_{2k+1}),
& t\in[t_{2k+1},t_{2k+2}),
\end{cases}
\end{equation}
and
\begin{equation}\label{e:xo2_time_development}
\tilde x_o^{(2)}(t)=
\begin{cases}
e^{A_o^{(2)}(t-t_{2k})}\,\tilde x_o^{(2)}(t_{2k}),
& t\in[t_{2k},t_{2k+1}),\\[2mm]
e^{\Lambda_{o,k}^{(2)}(t-t_{2k+1})}\,\tilde x_o^{(2)}(t_{2k+1}),
& t\in[t_{2k+1},t_{2k+2}).
\end{cases}
\end{equation}

Define the logarithmic norm (matrix measure) associated with the induced norm as
\begin{equation}\label{e:mu_o_i_def}
\mu^{(i)}
\coloneqq
\lambda_{\max}\!\left(\frac{A_o^{(i)}+\big(A_o^{(i)}\big)^\top}{2}\right),
\qquad i\in\{1,2\}.
\end{equation}

Using \Cref{d:log_norm}, the transition matrices corresponding to the observable components evolving without correction can be bounded, for all $t \ge t_0$, as
\begin{equation}\label{e:exp_bound_Ao_i}
\big\|e^{A_o^{(i)}(t-t_0)}\big\| \le e^{\mu^{(i)}(t-t_0)}.
\end{equation}
We next establish an exponential bound for the observable components on intervals where the corresponding channel is active. To this end, we require the following lemmas.

\begin{lemma}[Shift invariance of observability]\label{t:shift_observability}
Let $A \in \mathbb{R}^{n\times n}$ and $C \in \mathbb{R}^{p\times n}$.
If the pair $(A,C)$ is observable, then for any scalar $\eta \in \mathbb{R}$ the shifted pair
$(A+\eta I,\,C)$ is also observable.
\end{lemma}

\begin{proof}
See the Appendix.
\end{proof}

%%%%%%%%%%%%%%%%%%%%%%%%%%%%%%%%%%%%%%%%%%%%%%%%%%%%%%%%%%%%%%%
\smallskip
\begin{lemma}[Uniform exponential Euclidean bound]
\label{t:uniform-euclidean-bound}
Let $\Lambda\in\mathbb{R}^{n\times n}$ be Hurwitz. Then for any
$0<\eta<-\alpha(\Lambda)$, where
\begin{equation}
\alpha(\Lambda)
\coloneqq
\max_j \Re\big(\lambda_j(\Lambda)\big),
\end{equation}
there exists a constant $\kappa\ge 1$ such that the solution of
$\dot{\tilde{x}}(t)=\Lambda\tilde{x}(t)$ satisfies
\begin{equation}
\|\tilde{x}(t)\|
\le
\kappa\,e^{-\eta (t-t_0)}\,\|\tilde{x}(t_0)\|,
\qquad t\ge t_0.
\end{equation}
\end{lemma}

\smallskip
\begin{proof}
See the Appendix.
\end{proof}

%%%%%%%%%%%%%%%%%%%%%%%%%%%%%%%%%%%%%%%%%%%%%%%%%%%%%%%%%%%%%%%%%%%%%%%%%%%%%%%%%%%%%%%%%%%%%%%%%%%%%%%%%%%%%%%%%%%%%%%%%%%%%%%%%%%%%%%%%%%%%%%%%%%%%%%%%%%
\smallskip
\begin{lemma}[Observer gain with prescribed decay rate]
\label{t:observer-euclidean-bound}
Let $(A,C)$ be observable. Then for any desired $\eta>0$, there exists a gain
$L$ such that $\Lambda \coloneqq A-LC$ is Hurwitz and satisfies
$\alpha(\Lambda)<-\eta$. Consequently, the estimation-error dynamics
$\dot{\tilde{x}}(t)=\Lambda\tilde{x}(t)$ satisfy
\begin{equation}
\|\tilde{x}(t)\|
\le
\kappa\,e^{-\eta (t-t_0)}\,\|\tilde{x}(t_0)\|,
\qquad t\ge t_0,
\end{equation}
for some $\kappa\ge1$.
\end{lemma}

\begin{proof}
See the Appendix.
\end{proof}

By \Cref{t:observer-euclidean-bound}, for each $i \in \{1,2\}$, the gain $L_{o,k}^{(i)}$ can be chosen such that $\Lambda_{o,k}^{(i)}$ is Hurwitz with $\alpha\!\big(\Lambda_{o,k}^{(i)}\big)<-\eta_k^{(i)}$. Consequently, the corresponding transition matrices satisfy the Euclidean exponential bound
\begin{equation}\label{e:Lambda_ok_exp_bound}
\big\|e^{\Lambda_{o,k}^{(i)}(t-t_0)}\big\|
\le
\kappa_k^{(i)}\,e^{-\eta_k^{(i)}(t-t_0)},
\qquad t \ge t_0,\quad i \in \{1,2\},
\end{equation}
for some constants $\kappa_k^{(i)} \ge 1$.

To facilitate the design of the Luenberger gains, we introduce scalar envelope signals $z_1:\mathbb{R}_{\ge 0}\to\mathbb{R}_{\ge 0}$ and
$z_2:\mathbb{R}_{\ge 0}\to\mathbb{R}_{\ge 0}$ that upper bound the observable estimation-error norms:
\begin{equation}\label{e:composite_bounds}
\|\tilde x_o^{(1)}(t)\| \le z_1(t),
\qquad
\|\tilde x_o^{(2)}(t)\| \le z_2(t),
\qquad t\ge 0.
\end{equation}
The envelopes are defined to propagate on each subinterval of the $k$th switching cycle in accordance with the bounds
\eqref{e:exp_bound_Ao_i} and \eqref{e:Lambda_ok_exp_bound}.
In particular, over $[t_{2k},t_{2k+2})$,  we define
\begin{equation}\label{e:z1_evolution}
{\small
z_1(t):=
\begin{cases}
\displaystyle
\kappa_{k}^{(1)}\, e^{-\eta_{k}^{(1)}(t - t_{2k})}\, z_1(t_{2k}),
& t \in [t_{2k},\, t_{2k+1}), \\[6pt]
\displaystyle
e^{\mu^{(1)}(t - t_{2k+1})}\, z_1(t_{2k+1}),
& t \in [t_{2k+1},\, t_{2k+2}),
\end{cases}
}
\end{equation}
and
\begin{equation}\label{e:z2_evolution}
{\small
z_2(t):=
\begin{cases}
\displaystyle
e^{\mu^{(2)}(t - t_{2k})}\, z_2(t_{2k}),
& t \in [t_{2k},\, t_{2k+1}), \\[6pt]
\displaystyle
\kappa_{k}^{(2)}\, e^{-\eta_{k}^{(2)}(t - t_{2k+1})}\, z_2(t_{2k+1}),
& t \in [t_{2k+1},\, t_{2k+2}),
\end{cases}
}
\end{equation}

Since consecutive time segments with the same available output can be merged, we may represent the schedule as an alternating sequence over intervals $[t_{2k},t_{2k+2})$ with nonuniform lengths. Define
$\tau_k\coloneqq t_{2k+2}-t_{2k}$ and
$s_k\coloneqq (t_{2k+1}-t_{2k})/\tau_k$.
Then $y_1$ is available on $[t_{2k},t_{2k+1})$ for $s_k\tau_k$ and $y_2$ is available on
$[t_{2k+1},t_{2k+2})$ for $(1-s_k)\tau_k$.
Introduce the sampled bound vector
%%%
\begin{equation}\label{e:Zk_def}
Z_k \coloneqq
\begin{bmatrix}
z_1(t_{2k})\\[2pt]
z_2(t_{2k})
\end{bmatrix}.
\end{equation}
%%%
Evaluating \eqref{e:z1_evolution}--\eqref{e:z2_evolution} at $t=t_{2k+1}$ and $t=t_{2k+2}$ yields the linear recursion
%%%
\begin{equation}\label{e:Fz_system}
Z_{k+1}=F_k\,Z_k,
\end{equation}
%%%
where

\begin{equation}\label{e:Fk}
F_k
=
\begin{bmatrix}
\kappa_{k}^{(1)}\, e^{\beta_{1,k}(s_k)\,\tau_k} & 0\\[4pt]
0 & \kappa_{k}^{(2)}\, e^{\beta_{2,k}(s_k)\,\tau_k}
\end{bmatrix},
\end{equation}

in which,
\begin{equation}\label{e:beta1_def}
\beta_{1,k}(s_k)\coloneqq \mu^{(1)}(1-s_k)-\eta_k^{(1)} s_k,
\end{equation}
and
\begin{equation}\label{e:beta2_def}
\beta_{2,k}(s_k)\coloneqq \mu^{(2)} s_k-\eta_k^{(2)}(1-s_k).
\end{equation}

The observer design objective is thus achieved if the Luenberger gains are desined to make the discrete-time dynamical system in \eqref{e:Fk} asymptotically stable.

\section{Geometric Property of the Orthogonal Decomposition}\label{s:Orthogonal_Decom}

The following lemma records a basic geometric property that will be used repeatedly.

\smallskip
\begin{lemma}[Block norm preservation]\label{t:norm_preserve_blocks}
Let $T\in\mathbb R^{n\times n}$ be an orthogonal matrix and partition it as
$T=[\,T_o\ \ T_u\,]$, where $T_o\in\mathbb R^{n\times r}$ and
$T_u\in\mathbb R^{n\times (n-r)}$. Then, for any estimation-error
components $\tilde x_o\in\mathbb R^{r}$ and $\tilde x_u\in\mathbb R^{n-r}$,
\begin{equation}\label{e:Norm_Preserve}
\|T_o \tilde x_o\|=\|\tilde x_o\|,
\qquad
\|T_u \tilde x_u\|=\|\tilde x_u\|.
\end{equation}
\end{lemma}

\begin{proof}
See the Appendix.
\end{proof}

%%%%%%%%%%%%%%%%%%%%%%%%%%%%%%%%%%%%%%%%%%%%%%%%%%
\smallskip
\begin{lemma}[Stacked-output observability matrix]\label{t:O_stack}
Let  $C_1\in\mathbb R^{p_1\times n}$ and $C_2\in\mathbb R^{p_2\times n}$ and define
$C=\begin{bmatrix}C_1^\top & C_2^\top\end{bmatrix}^\top$.
Consider observability matrices $\mathcal O\coloneqq\mathcal O(A,C)$ and $\mathcal O_i \coloneqq \mathcal O(A,C_i), \quad i=1,2$. 
Then there exists a permutation matrix $M \in \mathbb{R}^{(p_1+p_2)n\times n}$ such that 
\begin{equation}\label{e:O_stack_identity}
\mathcal O
=M
\begin{bmatrix}
\mathcal O_1\\
\mathcal O_2
\end{bmatrix}.
\end{equation}
Consequently, if $\operatorname{rank}(\mathcal O)=n$ and $r_i \coloneqq\operatorname{rank}(\mathcal O_i)$, then
\begin{equation}\label{e:r1_r2_ge_n_from_O}
r_1+r_2\ge n.
\end{equation}
\end{lemma}

\begin{proof}
See the Appendix.
\end{proof}
%%%%%%%%%%%%%%%%%%%%%%%%%%%%%%%%%%%%%%%%%%%%%%%%%
\begin{theorem}    
[Error bound via observable parts]\label{t:xtilde_bound}
Consider the estimation error $\tilde x(t)\in\mathbb R^{n}$ and two orthogonal similarity transformations
$T_1,T_2\in\mathbb R^{n\times n}$ such that
\begin{equation}\label{e:xtilde_T1_T2}
\tilde{x}(t)=T_1\,\tilde{x}^{(1)}(t)=T_2\,\tilde{x}^{(2)}(t).
\end{equation}
Let the corresponding observability decompositions be
\begin{equation}\label{e:T_partition}
T_i=\begin{bmatrix}T_{io} & T_{iu}\end{bmatrix},\qquad
\tilde x^{(i)}(t)=
\begin{bmatrix}
\tilde x_o^{(i)}(t)\\
\tilde x_u^{(i)}(t)
\end{bmatrix},\quad i\in\{1,2\}.
\end{equation}
Assume that:
\begin{enumerate}
\item[(i)] $(A,C_i)$ is unobservable for each $i\in\{1,2\}$, and
\item[(ii)] $(A,C)$ is observable with $C=\begin{bmatrix}C_1^\top & C_2^\top\end{bmatrix}^\top$.
\end{enumerate}
Then there exists a constant $\alpha\ge 1$ such that, for all $t\ge 0$,
%%%
\begin{equation}\label{e:xtilde_bound_by_two_observable_parts}
\|\tilde x(t)\|
\le
\alpha\Big(\|\tilde x_o^{(1)}(t)\|+\|\tilde x_o^{(2)}(t)\|\Big).
\end{equation}
%%%
\end{theorem}
\smallskip
\begin{proof}
Define the observable-coordinate extractors
\begin{equation*}
S_{io}\coloneqq\begin{bmatrix}I & 0\end{bmatrix}T_i^\top = T_{io}^\top,
\qquad i\in\{1,2\},
\end{equation*}
so that $\tilde x_o^{(i)}(t)=S_{io}\tilde x(t)$.

\smallskip
Step 1 (kernel properties).
By assumption (i), $(A,C_i)$ is unobservable for each $i\in\{1,2\}$, hence
$\operatorname{rank}(\mathcal O_i)=:r_i<n$ for $i\in\{1,2\}$.
By \Cref{t:O_stack} and assumption (ii), $(A,C)$ is observable, i.e.,
$\operatorname{rank}(\mathcal O)=n$, and therefore
\begin{equation}\label{e:r1_r2_ge_n}
r_1+r_2\ge n.
\end{equation}

Moreover, by assumption (ii), for any $v\neq 0$ there exists $k\ge 0$ such that $CA^k v\neq 0$, i.e.,
\begin{equation*}
\big(\forall k\ge0,\ CA^k v=0\big)\ \Longrightarrow\ v=0.
\end{equation*}
Now let $v\in\ker(S_{1o})\cap\ker(S_{2o})$. Then $S_{1o}v=0$ and $S_{2o}v=0$.
In the orthogonal observability decomposition associated with $(A,C_i)$, the condition $S_{io}v=0$
implies $C_iA^k v=0$ for all $k\ge 0$. Hence $C_1A^k v=0$ and $C_2A^k v=0$ for all $k\ge0$, and thus
$CA^k v=0$ for all $k\ge0$. By the implication above, $v=0$. Therefore,
\begin{equation*}
\ker(S_{1o})\cap\ker(S_{2o})=\{0\}.
\end{equation*}

Step 2 (main bound).
Let
\begin{equation*}
H\coloneqq\begin{bmatrix}S_{1o}\\ S_{2o}\end{bmatrix},
\qquad
H\tilde x(t)=
\begin{bmatrix}
\tilde x_o^{(1)}(t)\\
\tilde x_o^{(2)}(t)
\end{bmatrix}.
\end{equation*}
Since $\ker(S_{1o})\cap\ker(S_{2o})=\{0\}$, we have $\ker(H)=\{0\}$ and hence $H$ has full column rank.
Therefore $\rho_{\min}(H)>0$. Define
\begin{equation}\label{e:alpha_def}
\alpha\coloneqq\frac{1}{\rho_{\min}(H)}.
\end{equation}
Then, for all $t\ge 0$,
\begin{equation*}
\|\tilde x(t)\|
\le
\alpha \left\|
\begin{bmatrix}
\tilde x_o^{(1)}(t)\\[1mm]
\tilde x_o^{(2)}(t)
\end{bmatrix}
\right\|
\le
\alpha\Big(\|\tilde x_o^{(1)}(t)\|+\|\tilde x_o^{(2)}(t)\|\Big),
\end{equation*}
which proves \eqref{e:xtilde_bound_by_two_observable_parts}.

\smallskip
Step 3 (show $\alpha\ge 1$).
Since $\ker(S_{2o})\neq\{0\}$, choose $v\in\ker(S_{2o})$ with $\|v\|=1$.
Because $T_1$ is orthogonal, the rows of $S_{1o}=T_{1o}^\top$ are orthonormal and hence $\|S_{1o}\|=1$.
Therefore,
\begin{equation*}
\|Hv\|
=
\left\|
\begin{bmatrix}
S_{1o}v\\
0
\end{bmatrix}
\right\|
=
\|S_{1o}v\|
\le 1.
\end{equation*}
Thus
\begin{equation*}
\rho_{\min}(H)=\min_{\|x\|=1}\|Hx\|
\le
\|Hv\|
\le
1,
\end{equation*}
which implies $\alpha=1/\rho_{\min}(H)\ge 1$.
\end{proof}

\section{ Main Results }\label{s:Main_Result}

\begin{theorem}[Estimation-error convergence]\label{t:main}
Consider the plant \eqref{e:system} under the alternating schedule
\eqref{e:switching_instants}, and the state-estimate dynamics
\eqref{e:switched_observer} with estimation error \eqref{e:est_error}.
Suppose that Assumptions~A1 and~A2 hold, and let the observer gain be given by
\eqref{e:L_sigma_cycle}.

Consider the envelope trajectories $z_1(\cdot)$ and $z_2(\cdot)$ defined in
\eqref{e:z1_evolution}, \eqref{e:z2_evolution}.
Let $\eta_k^{(1)}>0$ and $\eta_k^{(2)}>0$ be design parameters.
If, for some fixed $\delta_1,\delta_2\in(0,1)$,
\begin{equation}\label{e:eta1_bound}
\eta_k^{(1)}
\ \ge\
\frac{\mu^{(1)}(1-s_k)}{s_k}
\;+\;
\frac{1}{s_k\,\tau_k}\,
\ln\!\left(\frac{\kappa^{(1)}(\eta_k^{(1)})}{\delta_1}\right),
\end{equation}
and
\begin{equation}\label{e:eta2_bound}
\eta_k^{(2)}
\ \ge\
\frac{\mu^{(2)}\,s_k}{1-s_k}
\;+\;
\frac{1}{(1-s_k)\,\tau_k}\,
\ln\!\left(\frac{\kappa^{(2)}(\eta_k^{(2)})}{\delta_2}\right),
\end{equation}
hold for all $k\in\mathbb{Z}_{\ge0}$,
then the estimation error satisfies
\begin{equation}\label{e:xtilde_to_zero_main}
\|\tilde{x}(t)\|\to 0
\quad \text{as}\quad t\to\infty .
\end{equation}
\end{theorem}
%%%%%%%%%%%%%%%%%%%%%%%%%%%%%%%%%%%%%%%%%%%%%%%%%%

%%%%%%%%%%%%%%%%%%%%%%%%%%%%%%%%%%%%%%%%%%%%%%%%%%%%%%%%%%%%%%%%%%%%%%%%%%%%%%%%%%%%%%%%%%%%%%%%%%%%%%%%%%%%%%%%%%%%%%%%%%%%%%%%%%%%%%%%%

The following proof establishes \Cref{t:main}. The argument proceeds in two stages.
First, the envelope recursions derived in \Cref{s:Switch_Obs} yield a cycle-wise contraction for the sampled bound vector $Z_k$.
Second, the within-cycle envelope evolutions \eqref{e:z1_evolution}--\eqref{e:z2_evolution} propagate the decay from the switching instants to all times.
Finally, \Cref{t:xtilde_bound} converts decay of the observable envelopes into decay of the full estimation error.

\smallskip
\begin{proof}[\Cref{t:main}]
Consider the recursion \eqref{e:Fz_system} associated with the envelope trajectories
$z_1(\cdot)$ and $z_2(\cdot)$ defined in \eqref{e:z1_evolution}--\eqref{e:z2_evolution},
where $F_k$ is given by \eqref{e:Fk}.

\smallskip
\textit{Step 1 (use the $\eta$-inequalities to show $f_{11,k}\le \delta$ and $f_{22,k}\le \delta$).}
From \eqref{e:eta1_bound},
\begin{equation}\label{e:eta1_rearrange_main}
\eta_k^{(1)} s_k\tau_k
\ge
\mu^{(1)}(1-s_k)\tau_k
+
\ln\!\left(\frac{\kappa_k^{(1)}}{\delta_1}\right).
\end{equation}
Rearranging gives
\begin{equation}\label{e:beta1_bound_main}
\big(\mu^{(1)}(1-s_k)-\eta_k^{(1)} s_k\big)\tau_k
\le
-\ln\!\left(\frac{\kappa_k^{(1)}}{\delta_1}\right).
\end{equation}
Exponentiating \eqref{e:beta1_bound_main} and multiplying by $\kappa_k^{(1)}$ yields
\begin{equation}\label{e:f11_le_delta_main}
f_{11,k}
=
\kappa_k^{(1)}\exp\!\Big(\big(\mu^{(1)}(1-s_k)-\eta_k^{(1)} s_k\big)\tau_k\Big)
\le
\delta_1.
\end{equation}
Similarly, from \eqref{e:eta2_bound},
\begin{equation}\label{e:eta2_rearrange_main}
\eta_k^{(2)} (1-s_k)\tau_k
\ge
\mu^{(2)} s_k\tau_k
+
\ln\!\left(\frac{\kappa_k^{(2)}}{\delta_2}\right),
\end{equation}
which implies
\begin{equation}\label{e:f22_le_delta_main}
f_{22,k}
=
\kappa_k^{(2)}\exp\!\Big(\big(\mu^{(2)} s_k-\eta_k^{(2)} (1-s_k)\big)\tau_k\Big)
\le
\delta_2.
\end{equation}

%%%%%%%%%%%%%%%%%%%%%%%%%%%%%%%%%%%%%%%%%%%%%%%%%%%%%%%%%%%%%%%%%%%%%%%%%%%%%%%%%%%%%%%%%%%%%%%%%%%%

\textit{Step 2 (cycle-wise contraction of $Z_k$).}
From Step 1 we have
$f_{11,k}\le \delta_1$ and $f_{22,k}\le \delta_2$.
Define
\begin{equation}\label{e:delta_def_main}
\delta \coloneqq \max\{\delta_1,\delta_2\}\in(0,1).
\end{equation}
Since $F_k$ is diagonal with nonnegative diagonal entries, it follows that
\begin{equation}\label{e:Fk_norm_bound_main}
\|F_k\|
=
\max\{f_{11,k},f_{22,k}\}
\le
\delta,
\qquad \forall k\in\mathbb Z_{\ge 0}.
\end{equation}
Therefore,
\begin{equation}\label{e:Z_step_contraction_main_clean}
\|Z_{k+1}\|
=
\|F_k Z_k\|
\le
\|F_k\|\,\|Z_k\|
\le
\delta\,\|Z_k\|,
\qquad \forall k\in\mathbb Z_{\ge 0}.
\end{equation}
Iterating \eqref{e:Z_step_contraction_main_clean} gives
\begin{equation}\label{e:Zk_to_zero_main}
\|Z_k\|
\le
\delta^k \|Z_0\|
\ \longrightarrow\ 0
\quad \text{as}\quad k\to\infty.
\end{equation}
Since
$Z_k=\begin{bmatrix}z_1(t_{2k}) & z_2(t_{2k})\end{bmatrix}^\top$,
\eqref{e:Zk_to_zero_main} implies
$z_1(t_{2k})\to 0$ and
$z_2(t_{2k})\to 0$ as $k\to\infty$.

%%%%%%%%%%%%%%%%%%%%%%%%%%%%%%%%%%%

%%%%%%%%%%%%%%%%%%%%%%%%%%%%%%%
\smallskip
\textit{Step 3 (from $z_i(t_{2k})\to 0$ to $z_i(t)\to 0$).}
Fix $k\in\mathbb Z_{\ge 0}$ and let $t\in[t_{2k},t_{2k+2})$.
From \eqref{e:z1_evolution}, if $t\in[t_{2k},t_{2k+1})$ then
\[
z_1(t)
=
\kappa_{k}^{(1)} e^{-\eta_{k}^{(1)}(t - t_{2k})} z_1(t_{2k})
\le
\kappa_{k}^{(1)} z_1(t_{2k}),
\]
and if $t\in[t_{2k+1},t_{2k+2})$ then $t-t_{2k+1}\le t_{2k+2}-t_{2k+1}\le \tau_k$, so
\[
z_1(t)
=
e^{\mu^{(1)}(t - t_{2k+1})} z_1(t_{2k+1})
\le
e^{\mu^{(1)}\tau_k}\, z_1(t_{2k+1}).
\]
Moreover, evaluating the first branch of \eqref{e:z1_evolution} at $t=t_{2k+1}$ gives
$
z_1(t_{2k+1})
=
\kappa_{k}^{(1)} e^{-\eta_{k}^{(1)}(t_{2k+1}-t_{2k})} z_1(t_{2k})
\le
\kappa_{k}^{(1)} z_1(t_{2k}),
$
and therefore, for all $t\in[t_{2k},t_{2k+2})$,
\begin{equation}\label{e:z1_cycle_gamma_bound}
z_1(t)\le \gamma_{1,k}\, z_1(t_{2k}),
\qquad
\gamma_{1,k} \coloneqq \kappa_{k}^{(1)} e^{\mu^{(1)}\tau_k}.
\end{equation}

Similarly, from \eqref{e:z2_evolution}, for all $t\in[t_{2k},t_{2k+2})$,
\begin{equation}\label{e:z2_cycle_gamma_bound}
z_2(t)\le \gamma_{2,k}\, z_2(t_{2k}),
\qquad
\gamma_{2,k}\coloneqq \kappa_{k}^{(2)} e^{\mu^{(2)}\tau_k}.
\end{equation}

Under \textbf{(A2)} we have $ \underline{\tau}<\tau_k\le \bar\tau$, and the design yields bounded
$\{\kappa_k^{(1)}\}$ and $\{\kappa_k^{(2)}\}$.
Hence there exist finite constants $\bar\gamma_1,\bar\gamma_2$ such that
\begin{equation}\label{e:gamma_bar_bounds}
\gamma_{1,k}\le \bar\gamma_1,
\qquad
\gamma_{2,k}\le \bar\gamma_2,
\qquad \forall k\in\mathbb Z_{\ge 0}.
\end{equation}
Therefore, \eqref{e:z1_cycle_gamma_bound}--\eqref{e:gamma_bar_bounds} together with
$z_1(t_{2k})\to 0$ imply $z_1(t)\to 0$ as $t\to\infty$.
Likewise, \eqref{e:z2_cycle_gamma_bound}--\eqref{e:gamma_bar_bounds} and $z_2(t_{2k})\to 0$
imply $z_2(t)\to 0$ as $t\to\infty$.

\smallskip
\textit{Step 4 (convergence of $\|\tilde x(t)\|$).}
By \Cref{t:xtilde_bound}, there exists $\alpha\ge 1$ such that
\begin{equation}\label{e:xtilde_alpha_z_main}
\begin{aligned}
\|\tilde x(t)\|
&\le
\alpha\Big(\|\tilde x_o^{(1)}(t)\|+\|\tilde x_o^{(2)}(t)\|\Big)\\
&\le
\alpha\big(z_1(t)+z_2(t)\big),
\qquad \forall t\ge 0 .
\end{aligned}
\end{equation}

Since Step~3 shows that $z_1(t)\to 0$ and $z_2(t)\to 0$ as $t\to\infty$, we have
$z_1(t)+z_2(t)\to 0$.
Therefore, \eqref{e:xtilde_alpha_z_main} implies $\|\tilde x(t)\|\to 0$ as $t\to\infty$.
\end{proof}

%%%%%%%%%%%%%%%%%%%%%%%%%%%%%%%%%%%%%%%%%%%%%%%%%%%%%%%%%%%%%%%%%%%%%%%%%%%%%%%%%%%
\smallskip
The following lemma indicates that the hypotheses of \Cref{t:main} are not restrictive.

\smallskip
\begin{lemma}[Feasibility of the design inequalities]\label{t:feasibility_eta}
For each cycle $k$, the inequalities
\eqref{e:eta1_bound} and \eqref{e:eta2_bound} are feasible.
\end{lemma}

\begin{proof}
See the Appendix.
\end{proof}

\section{LMI-based Luenberger Gain Design}
\label{s:LMI_Gain}

This section presents a systematic procedure to synthesize the Luenberger gains
$L_o^{(1)}$ and $L_o^{(2)}$ to satisfy hypotheses \eqref{e:eta1_bound} and \eqref{e:eta2_bound} of \Cref{t:main} 
Each switching cycle is characterized by a pair $(\tau_k,s_k)$ as in \eqref{e:tau_k} and \eqref{e:s_k}.
For each distinct pair $(\tau,s)$, we synthesize two observable gains
$L_o^{(1)}(\tau,s)$ and $L_o^{(2)}(\tau,s)$ (one per output channel), which can be reused in any
cycle satisfying $(\tau_k,s_k)=(\tau,s)$.
These gains are then embedded into the full-state coordinates via
\eqref{e:L_transformed}--\eqref{e:L_original} to obtain the channel-dependent
gain pair $(L_{1,k},L_{2,k})$.

Recall that for each output channel $i\in\{1,2\}$, the observable estimation error evolves according to
\begin{equation}
\dot{\tilde{x}}_o^{(i)}(t)
=
\big(A_o^{(i)}-L_o^{(i)}C_o^{(i)}\big)\tilde{x}_o^{(i)}(t),
\qquad t\ge t_0,
\end{equation}
and we define the corresponding closed-loop matrix
\begin{equation}
\Lambda^{(i)}
\coloneqq
A_o^{(i)}-L_o^{(i)}C_o^{(i)}.
\end{equation}
The goal is to design $L_o^{(i)}$ such that $\Lambda^{(i)}$ is Hurwitz with a prescribed decay rate.

In the following , we drop the channel index (i) to simplify notation, since the synthesis steps are identical for $i \in \{1,2\}$.

\subsection{Computation of gains for a given decay rate}

\begin{lemma}[LMI synthesis of observer gain]
\label{t:observer-sdp-design}
Let $(A,C)$ be observable and fix $\eta>0$.
Consider the semidefinite program
\begin{align}
\min_{P,Y,\theta} \quad & \theta \label{e:SDP_1}\\
\text{subject to}\quad
& P=P^\top\succ0, \label{e:SDP_2}\\
& A^\top P+PA-C^\top Y^\top-YC+2\eta P \preceq 0, \label{e:SDP_3}\\
& I \preceq P \preceq \theta I. \label{e:SDP_4}
\end{align}
Let $(P^\star,Y^\star,\theta^\star)$ be an optimal solution and define
\begin{equation}\label{e:Lstar}
L^\star \coloneqq (P^\star)^{-1}Y^\star.
\end{equation}
Then the estimation error satisfies
\begin{equation}\label{e:error_bound}
\|\tilde{x}(t)\|
\le
\kappa^\star\,e^{-\eta (t-t_0)}\|\tilde{x}(t_0)\|,
\qquad t\ge t_0.
\end{equation}
with $\kappa^\star=\sqrt{\theta^\star}$.
\end{lemma}

%%%%%%%%%%%%%%%%%%%%%%%%%%%%%%%%%%%%%%%%%%%%%%%%%%%%%%%%%%%%%%%%%%%%%%%%%%%%%%%%%%%%%%%%%%%%
\noindent
For each channel $i\in\{1,2\}$, apply \Cref{t:observer-sdp-design}
to the observable pair $(A_o^{(i)},C_o^{(i)})$ with a prescribed rate $\eta^{(i)}>0$.
This yields gains $L_o^{(i)}$ such that
\begin{equation}
\|\tilde{x}_o^{(i)}(t)\|
\le
\kappa^{(i)}
e^{-\eta^{(i)}(t-t_0)}
\|\tilde{x}_o^{(i)}(t_0)\|,
\qquad t\ge t_0,
\end{equation}
for some $\kappa^{(i)}\ge \kappa^{\star(i)}$.
These constants directly enter the switching-cycle bounds derived in
\Cref{s:Switch_Obs}.

%%%%%%%%%%%%%%%%%%%%%%%%%%%%%%%%%%%%%%%%%%%%%%%%%%%%%%%%%%%%%%%%%%%%%

\subsection{Selection of decay rates}

To complete the observer design procedure, we must select gains that satisfy the inequalities in \eqref{e:eta1_bound} and \eqref{e:eta2_bound}. 
Note that the constant $\kappa^{(i)}$ obtained from the SDP in \eqref{e:SDP_1}--\eqref{e:SDP_4} depends on the chosen decay rate $\eta^{(i)}$. 
As a result, $\eta^{(i)}$ appears on both sides of the inequality, leading to an implicit design condition.

To characterize the minimal admissible decay rates, we define $\eta_{\min}^{(i)}$ as the solution of the fixed-point equations
\begin{equation}\label{e:eta1_fixed}
\eta_{\min}^{(1)}=
\frac{\mu^{(1)}(1-s)}{s}
\;+\;
\frac{1}{s\,\tau}\,
\ln\!\left(\frac{\kappa^{(1)}(\eta_{\min}^{(1)})}{\delta_1}\right),
\end{equation}
and
\begin{equation}\label{e:eta2_fixed}
\eta_{\min}^{(2)}=
\frac{\mu^{(2)}\,s}{1-s}
\;+\;
\frac{1}{(1-s)\,\tau}\,
\ln\!\left(\frac{\kappa^{(2)}(\eta_{\min}^{(2)})}{\delta_2}\right).
\end{equation}

By \Cref{t:feasibility_eta}, these fixed-point equations admit solutions.

In practice, the solutions are obtained numerically by evaluating the SDP 
\eqref{e:SDP_1}--\eqref{e:SDP_4} for candidate values of $\eta^{(i)}$, thereby computing 
$\kappa^{(i)}(\eta^{(i)})$ and identifying the smallest $\eta^{(i)}$ that satisfies the above equations.

Any choice $\eta^{(i)} > \eta_{\min}^{(i)}$ satisfies the inequalities 
\eqref{e:eta1_bound} and \eqref{e:eta2_bound}. 
Therefore, for given $(s,\tau)$, the observer design proceeds by first determining $\eta_{\min}^{(i)}$, and then selecting any $\eta^{(i)} > \eta_{\min}^{(i)}$ and solving the SDP \eqref{e:SDP_1}--\eqref{e:SDP_4} to obtain the gain matrices $\{L_{o}^{(i)}\}$.

\section{Numerical Experiments}\label{Sec_Sim}
This section illustrates the performance of the proposed switching observer on a representative mobile sensing scenario involving one camera and two moving agents.

\subsection{System Model and Sensing Architecture}\label{subsec:sim_model}

The considered sensing scenario is illustrated in \Cref{f:Camera_Agents}. 
A single moving camera aims to estimate the states of two moving agents in the plane. 
Due to field-of-view limitations, the camera can observe at most one agent at any given time.

The sensing architecture satisfies the following assumptions

\begin{enumerate}
    \item The camera continuously measures its own planar position $(p_{x,c},p_{y,c})$
    \item At any time, only one agent lies within the field of view, so the camera measures the planar position of exactly one agent
    \item The camera allocates its sensing time between the two agents, resulting in switching cycles characterized by $(\tau_k,s_k)$
\end{enumerate}

For each subsystem $j\in\{c,1,2\}$, let $p_j(t)\in\mathbb{R}^2$ and $v_j(t)\in\mathbb{R}^2$ denote position and velocity, and define the subsystem state as $\chi_j(t)=[\,p_j^\top(t)\; v_j^\top(t)\,]^\top\in\mathbb{R}^4$. 
The overall state is $x(t)=[\,\chi_c^\top(t)\; \chi_1^\top(t)\; \chi_2^\top(t)\,]^\top\in\mathbb{R}^{12}$.

Each subsystem is modeled as a double integrator with velocity matrix $D_j\in\mathbb{R}^{2\times2}$.
The overall dynamics are block diagonal
\begin{equation}\label{e:sim_A_blkdiag}
\dot{x}(t)=Ax(t), 
\qquad 
A=\mathrm{blkdiag}(A_c,A_1,A_2),
\end{equation}
where
\begin{equation}\label{e:sim_Aj}
A_j=
\begin{bmatrix}
0_{2\times 2} & I_2 \\
0_{2\times 2} & D_j
\end{bmatrix},
\qquad j\in\{c,1,2\}.
\end{equation}

At any time, the camera measures its own position and the position of exactly one agent. 
When channel $i\in\{1,2\}$ is active, the output is $y_i(t)=C_i x(t)$, where $y_i(t)=[\,p_c^\top(t)\; p_i^\top(t)\,]^\top\in\mathbb{R}^4$ and $C_i\in\mathbb{R}^{4\times 12}$ selects the corresponding position components.

%%%%%%%%%%%%%%%%%%%%%%%%%%%%%%%%%%

Only one of $y_1(t)$ or $y_2(t)$ is available at a time.
Individually, $(A,C_1)$ and $(A,C_2)$ are unobservable and not detectable, as
\begin{equation*}
\operatorname{rank}\big(\mathcal O(A,C_1)\big)=\operatorname{rank}\big(\mathcal O(A,C_2)\big)=8<12,
\end{equation*}
and the corresponding unobservable subspaces contain unstable modes.
In contrast, when both channels are available, the aggregated pair becomes observable:
\begin{equation*}
\operatorname{rank}\big(\mathcal O(A,C)\big)=12, 
\qquad 
C=\begin{bmatrix}C_1\\ C_2\end{bmatrix}.
\end{equation*}
This corresponds to the intermittent-output setting addressed in the present paper.

% Only one of $y_1(t)$ or $y_2(t)$ is available at a time.
% Individually, $(A,C_1)$ and $(A,C_2)$ are unobservable and not detectable, whereas access to both channels renders the overall pair observable. 
% This corresponds to the intermittent-output setting addressed in the present paper.

This simulation setup satisfies the structural assumptions required by \Cref{t:main}. 
The switching schedule $(\tau_k,s_k)$ is selected by the camera, and channel-dependent observer gains are synthesized using the LMI procedure of \Cref{s:LMI_Gain}. 
The next subsection specifies the numerical parameters and verifies that the resulting cycle multipliers satisfy the sufficient contraction conditions.
%%%%%%%%%%%%%%%%%%%%%%%%%%%%%%%%%%%%%%%%%%%%%%%%%%%%%%%%%%%%%%%%%%%%%%%%%%%%%%%%%%%%%%%%%%%%%%%%%%%%%%%%%%%%%%%%%%%%%%%%%%%%%%%%%%%%%%%%%%%%%%%%%%%
\subsection{Simulation Parameters}\label{subsec:sim_params}

In \eqref{e:sim_Aj}, we select diagonal velocity matrices $D_j=\mathrm{diag}(d_{x,j},d_{y,j})$ for $j\in\{c,1,2\}$. 
The camera dynamics use $D_c=\mathrm{diag}(0,0)$, which yields marginal double-integrator behavior. 
The agent dynamics use $D_1=\mathrm{diag}(0,0.2)$ and $D_2=\mathrm{diag}(0,0.1)$, so that each agent contains one unstable velocity mode corresponding to the positive diagonal entry. 

Two switching scenarios are considered. 
In the periodic case, the cycle parameters are constant with $(\tau_k,s_k)\equiv(2\,\mathrm{s},0.6)$ for all $k$. 
In the nonperiodic case, the cycle parameters vary over successive cycles according to $(\tau_k,s_k)\in\{(4\,\mathrm{s},0.5),(3\,\mathrm{s},0.6),(3\,\mathrm{s},0.6),(2\,\mathrm{s},0.4)\}$. 
For each selected pair $(\tau_k,s_k)$, the observer gains $L_1$ and $L_2$ are synthesized using the LMI procedure of \Cref{s:LMI_Gain}.

%%%%%%%%%%%%%%%%%%%%%%%%%%%%%%%%%%%%%%%%%%%%%%

%%%%%%%%%%%%%%%%%%%%%%%%%%%%%%%%%%%%%%%%%%%%%%%%%

\subsection{Observer Construction and Gain Synthesis}\label{subsec:sim_gain}

For each channel $i\in\{1,2\}$, an orthogonal observability decomposition is computed using an orthogonal similarity transformation $T_i$. 
This yields the observable subsystem matrix $A_o^{(i)}$ corresponding to the output channel $y_i$ and the associated reduced output matrix $C_o^{(i)}$ used in the observer design. 
The observer is implemented according to the switching structure described in \Cref{s:Switch_Obs}

The constants $\mu^{(i)}$ are computed as the logarithmic norms of the matrices $A_o^{(i)}$, as defined in \eqref{e:mu_o_i_def}. 
For the present system, this yields $\mu^{(1)}=0.6099$ and $\mu^{(2)}=0.5525$.

The decay rates are selected based on the fixed-point conditions \eqref{e:eta1_fixed}--\eqref{e:eta2_fixed}, which determine the minimal admissible values $\eta_{\min}^{(i)}$ ensuring feasibility of the design conditions.

For each cycle, a gain pair $(L_{1,k},L_{2,k})$ is designed offline for the selected cycle parameters $(\tau_k,s_k)$ using the LMI procedure of \Cref{s:LMI_Gain}. 
The resulting gain pair depends only on $(\tau_k,s_k)$ and can therefore be reused whenever the same cycle parameters occur.

\medskip

\noindent\textit{Periodic case:}
When $(\tau_k,s_k)\equiv(2\,\mathrm{s},0.6)$, solving \eqref{e:eta1_fixed}--\eqref{e:eta2_fixed} yields 
$(\eta_{\min}^{(1)},\eta_{\min}^{(2)})=(1.5,\,3.24)$. 
We select $(\eta_k^{(1)},\eta_k^{(2)})\equiv(1.6,3.5)$, which satisfy $\eta_k^{(i)}>\eta_{\min}^{(i)}$, and compute the corresponding gains once offline.

\medskip

\noindent\textit{Nonperiodic case:}
For $(\tau_k,s_k)\in\{(4\,\mathrm{s},0.5),(3\,\mathrm{s},0.6),(2\,\mathrm{s},0.4)\}$, the fixed-point equations are solved for each distinct configuration to obtain $(\eta_{\min}^{(1)},\eta_{\min}^{(2)})$. 
The selected decay rates satisfy $\eta^{(i)}>\eta_{\min}^{(i)}$, and the corresponding gains are computed via the LMI procedure and reused across cycles with identical parameters. 
The values are summarized in   \Cref{T:eta_values}.

%%%%%%%%%%%%%%%%%%%%%%%%%%%%%%%%%%%%%%%%%%

\begin{table}[h]
\centering
\caption{Minimal and selected decay rates for distinct cycle configurations}
\label{T:eta_values}
\begin{tabular}{c|c|c}
\hline
$(\tau, s)$ & $(\eta_{\min}^{(1)}, \eta_{\min}^{(2)})$ & $(\eta^{(1)},\eta^{(2)})$ \\
\hline
$(4\,\mathrm{s},\,0.5)$ & $(1.18,1.05)$ & $(1.2,\,1.1)$ \\
$(3\,\mathrm{s},\,0.6)$ & $(0.96,2.12)$ & $(1,\,2.2)$ \\
$(2\,\mathrm{s},\,0.4)$ & $(3.43, 1.35)$ & $(3.5,\,1.5)$ \\
\hline
\end{tabular}
\end{table}

%%%%%%%%%%%%%%%%%%%%%%%%%%%%%%%%%%%%%%
%%%%%%%%%%%%%%%%%%%% Figures %%%%%%%%%%%%%%%%%%%%%%%%%%%%%%%%%%%%%%%%%

\begin{figure}[t]
\centering
\includegraphics[width=0.40\columnwidth]{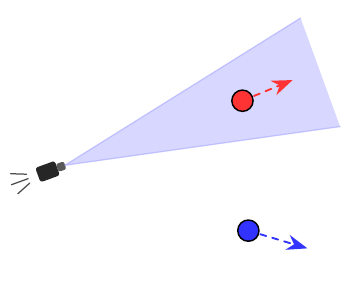}
\caption{Mobile camera sensing architecture with two output channels. Due to field-of-view constraints, the camera observes only one agent at a time}
\label{f:Camera_Agents}
\end{figure}

%%%%%%%%%%%%%%%%%%%%%%%%%%%%%%%%%%%%%%%%%%%%%%%%%%%%%%%%%%%%%%%%%%%%%%%%%%%%%%%%%%%%%%%%%%%%%%%%%%%%%%%%%%%%%%%%%%%%%%%%%%%%%%%%%%%%%%%%%%%%%%%%%%%%%%%%%%%%%%%%%%%%%%%%%%%%%%%%%%%%%%%%%%%%%%%%%%%%%%%%%%%%%
%%%%%%%%%%%%%%%%%%%%%%%%%%%%%%%%%%%%%%%%%%%%%%%%%%%%%%%%%%%%%%%%%%%%%%%%%%%%%%%%%%%%%%%%%%%%%%%%%%%%%%%%%%%%%%%%%%%%%%%%%%%%%%%%%%%%%%%%%%%%%%%%%%%%%%%%%%
%%%%%%%%%%%%%%%%%%%%%%%%%%%%%%%%%%%%%%%%%%%%%%%%%%%%%%%%%%%%%%%%%%%%%%%%%%%%%%%%%%%%%%%%%%%%%%%%%%%%%%%%%%%%%%%%%%%%%%%%%%%%%%%%%%%%%%%%%%%%%%%%%%%%%%%%%%
\subsection{Estimation Error Results}\label{subsec:sim_results}

We evaluate the estimation performance of the proposed switching observer under the periodic and nonperiodic sensing schedules.

For both scenarios, the observer is initialized with a perturbed estimate,
\[
\hat{x}(0)=x(0)+0.1\,\mathrm{randn}(\mathrm{size}(x(0))),
\]
so that the initial estimation error $\tilde{x}(0)=\hat{x}(0)-x(0)$ is nonzero.

In the periodic switching scenario, \Cref{f:error_periodic} illustrates the evolution of the estimation errors. 
The observable-coordinate error norms $\|\tilde{x}_o^{(1)}(t)\|$ and $\|\tilde{x}_o^{(2)}(t)\|$ both converge to zero, and consequently the overall estimation error norm $\|\tilde{x}(t)\|$ in the original coordinates also converges to zero. 
The switching signal $\rho(t)$ is included for reference.

In the nonperiodic switching scenario, \Cref{f:error_nonperiodic} illustrates the evolution of the estimation errors. 
The observable-coordinate error norms $\|\tilde{x}_o^{(1)}(t)\|$ and $\|\tilde{x}_o^{(2)}(t)\|$ both converge to zero, and consequently the overall estimation error norm $\|\tilde{x}(t)\|$ in the original coordinates also converges to zero despite the varying cycle parameters. 
The switching signal $\rho(t)$ is included for reference.

These results are consistent with the convergence claim of \Cref{t:main} for both periodic and nonperiodic schedules, even though each individual channel pair $(A,C_i)$ is not detectable.

\begin{figure}[t]
\centering
\includegraphics[width=0.8\linewidth]{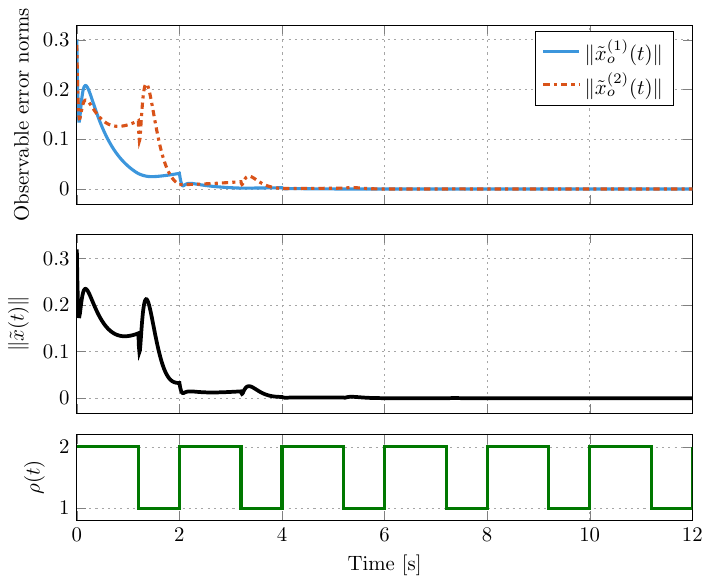}
\caption{Estimation error under nonperiodic switching. 
Top: observable-coordinate error norms $\|\tilde{x}^{(1)}_{o}(t)\|$ and $\|\tilde{x}^{(2)}_{o}(t)\|$. 
Middle: total estimation error $\|\tilde{x}(t)\|$ in the original coordinates. 
Bottom: switching signal $\rho(t)$.}
\label{f:error_periodic}
\end{figure}

\begin{figure}[t]
\centering
\includegraphics[width=0.8\linewidth]{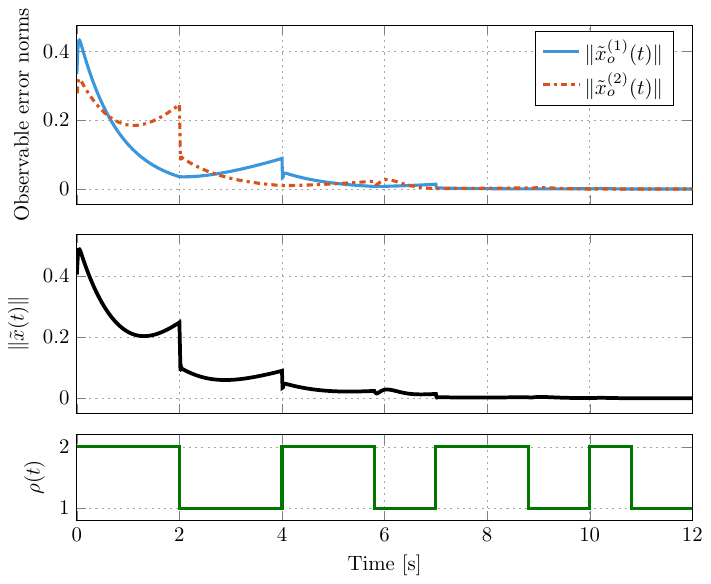}
\caption{Estimation error under nonperiodic switching. 
Top: observable-coordinate error norms $\|\tilde{x}^{(1)}_{o}(t)\|$ and $\|\tilde{x}^{(2)}_{o}(t)\|$. 
Middle: total estimation error $\|\tilde{x}(t)\|$ in the original coordinates. 
Bottom: switching signal $\rho(t)$.}
\label{f:error_nonperiodic}
\end{figure}

%%%%%%%%%%%%%%%%%%%%%%%%%%%%%%%%%%%%%%%%%%%%%%%%%%%
 % \vspace{-1em}

\subsection{Discussion}
The numerical results validate the theoretical developments of
\Cref{s:Orthogonal_Decom,s:Main_Result,s:LMI_Gain}    and illustrate how the cycle-wise contraction
mechanism translates into continuous-time decay of the full
estimation error. In particular, the simulations confirm that
although the observable estimation-error components may exhibit
temporary growth within a switching cycle, the net effect over each
cycle is strictly contractive under the conditions of \Cref{t:main}.

Recall that the sampled envelope vector
\[
Z_k =
\begin{bmatrix}
z_1(t_{2k})\\
z_2(t_{2k})
\end{bmatrix}
\]
satisfies the linear recursion
\[
Z_{k+1} = F_k Z_k,
\]
where $F_k$ is diagonal with entries
$f_{11,k}$ and $f_{22,k}$ determined by the decay rates
$\eta^{(i)}_k$, the logarithmic norms $\mu^{(i)}$, and the
cycle parameters $(\tau_k,s_k)$. The gain-selection inequalities
(79)--(80) guarantee the existence of a constant
$\delta \in (0,1)$ such that
\[
\|Z_{k+1}\| \le \delta \|Z_k\|.
\]
The periodic simulations show that the sampled envelope norms
decay geometrically, which is consistent with the bound
$\|Z_k\| \le \delta^k \|Z_0\|$. In the nonperiodic case,
the slopes of the logarithmic error curves vary from cycle to
cycle, reflecting the dependence of $F_k$ on $(\tau_k,s_k)$;
nevertheless, as long as the contraction inequalities are satisfied
uniformly over all cycles, convergence is preserved.

A characteristic feature observed in the simulations is the
within-cycle transient growth of the observable error component
associated with the inactive channel. When channel $i$ is not
active, its observable component evolves according to the open-loop
matrix $A_o^{(i)}$ and may grow at a rate governed by
$\mu^{(i)}$. However, once the corresponding channel becomes
active, the injected gain produces exponential decay at rate
$\eta^{(i)}_k$, and the design ensures that the net effect over the
full cycle satisfies
\[
\kappa^{(i)}_k
\exp\!\big(
\beta_{i,k}(s_k)\,\tau_k
\big)
\le \delta_i.
\]
The figures clearly
illustrate this alternating growth--decay pattern, while the overall
error norm decreases monotonically at the switching instants.

Finally, the simulations provide a concrete interpretation of the
finite-horizon informativeness condition. Neither $(A,C_1)$ nor
$(A,C_2)$ is individually observable; nevertheless, the combined
effect of accessing both outputs over each switching cycle
accumulates sufficient information to reduce the estimation error.
The geometric bound
\[
\|\tilde{x}(t)\|
\le
\alpha\big(z_1(t)+z_2(t)\big)
\]
confirms that decay of the channel-dependent observable components
implies decay of the full-state error. Thus, the numerical results
demonstrate that per-cycle information accumulation, together with
LMI-based gain synthesis, is sufficient to guarantee asymptotic
convergence even when no single output channel is observable or
detectable on its own.

\Cref{f:eta_min1,f:eta_min2} characterize the minimum admissible decay rates $\eta_{\min}^{(1)}$ and $\eta_{\min}^{(2)}$ as functions of the cycle parameters $(\tau,s)$, thereby defining the admissible design region. 
These plots provide a direct graphical alternative to solving the fixed-point equations \eqref{e:eta1_fixed}--\eqref{e:eta2_fixed}, allowing $(\eta_{\min}^{(1)},\eta_{\min}^{(2)})$ to be read off for any given $(\tau,s)$. 
The curves are shown for representative values $\tau=1,2,3,4,5$ and all $s\in(0,1)$, illustrating the dependence of the required decay rates on the switching schedule.

Selecting $\eta^{(i)}>\eta_{\min}^{(i)}$ strengthens contraction during the active subinterval and typically accelerates the asymptotic decay of $\|\tilde{x}(t)\|$. 
However, \Cref{f:kappa_LMI} shows that the associated  $\kappa^\star$ increases with $\eta$, enlarging the multiplicative prefactor in the error bound. 
Thus, admissible parameters must ensure that the combined cycle condition still yields a contraction factor $\delta<1$ as required by \Cref{t:main}.

The dependence of $\kappa^{(i)}$ on $\eta^{(i)}$ reveals a fundamental tradeoff: larger $\eta^{(i)}$ improves asymptotic decay but increases $\kappa^{(i)}$, thereby amplifying transient growth within each switching cycle. 
Such amplification is an inherent consequence of enforcing faster decay via LMI certificates.

%%%%%%%%%%%%%%%%%%%%%%%%%%%%%%%%%%%%%%%%%%%%%%%%%%%%%%%

\begin{figure}[t]
\centering
\includegraphics[width=0.8\linewidth]{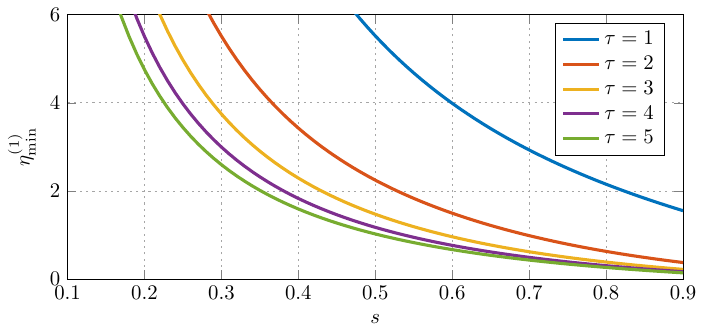}
\caption{Minimum decay rate $\eta_{\min}^{(1)}$ satisfying contraction conditions.}
\label{f:eta_min1}
\end{figure}

\begin{figure}[t]
\centering
\includegraphics[width=0.8\linewidth]{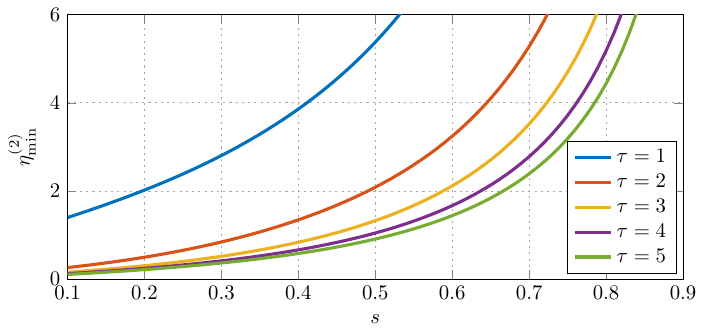}
\caption{Minimum decay rate $\eta_{\min}^{(2)}$ satisfying contraction conditions.}
\label{f:eta_min2}
\end{figure}

\begin{figure}[t]
\centering
\includegraphics[width=0.8\linewidth]{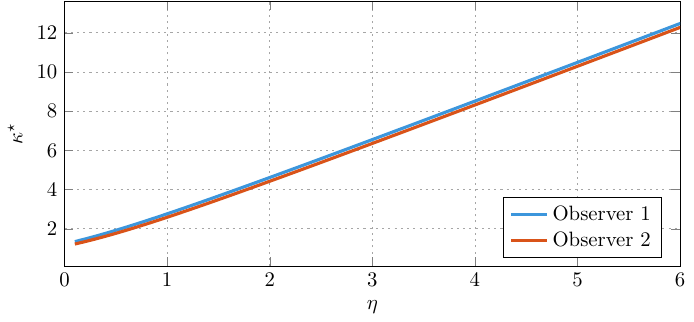}
\caption{Minimum achievable bound constant $\kappa^\star$ as a function of the decay rate $\eta$, obtained via LMI synthesis.}
\label{f:kappa_LMI}
\end{figure}

\FloatBarrier

\section{Conclusion}

In this paper, we developed a switching Luenberger observer
for linear time-invariant systems with intermittently available
outputs under a predefined output-switching rule. Convergence
is achieved by leveraging the observability of the aggregated
output together with uniform bounds on the switching intervals,
without imposing any detectability condition on the individual
channels. We design channel-dependent Luenberger gains via
linear matrix inequalities (LMIs) and prove that the estimation
error converges to zero for both periodic and nonperiodic
schedules characterized by $(\tau_k,s_k)$. The proposed analysis
exploits orthogonal transformations to identify channel-dependent
observable subspaces and to construct norm-based bounds that are
invariant under such coordinate changes.

In contrast to Riccati-based observer designs that require online, time-varying gain updates, the proposed approach enables offline LMI-based synthesis of observer gains for prescribed switching parameters $(\tau_k,s_k)$. This leads to constant gains for periodic schedules and a table-based gain selection mechanism for nonperiodic schedules, thereby reducing online computational complexity.

Several directions for future work arise from this study. First, the framework can be extended to incorporate measurement noise and process disturbances, leading to robustness guarantees and stochastic performance bounds. Second, the proposed methodology can be generalized to systems with more than two output channels, where only subsets of outputs are available at each time or where partial access to multiple outputs occurs simultaneously. Finally, extending the switching observer design to nonlinear systems, possibly via local linearization, contraction analysis, or incremental stability tools, remains an important direction for broadening the applicability of the results.

\newpage

\section{Appendices}

%%%% Auxiliary lemmas
\subsection{Auxiliary Lemmas}
The following auxiliary lemmas are used in the appendix to establish the main results.

\smallskip
\noindent\textbf{Lemma A.1.}\par

Let $c_1>0$ and $c_2>0$. Then the inequality
\begin{equation*}
\eta \;\ge\; c_1 + c_2 \ln \eta
\end{equation*}
admits a solution.

\begin{proof}
Define $f(\eta):=\eta - c_2\ln\eta$. Then
\[
f'(\eta)=1-\frac{c_2}{\eta},
\]
so $f'(\eta)>0$ for all $\eta>c_2$, and hence $f$ is strictly increasing on $(c_2,\infty)$. Moreover, $f(\eta)\to\infty$ as $\eta\to\infty$. Therefore, there exists $\eta>c_2$ such that $f(\eta)\ge c_1$, which proves the claim.
\end{proof}
\smallskip

%%%%%%%%%%%%%%%%%%%%%%%%%%%%%%%%%%%%%%%%%%%%%%%%%%%%%%%
\noindent\textbf{Lemma A.2.} 
(\textit{Polynomial decay bound.})\par
\smallskip
% \label{t:polynomial-prefactor}
Let $(A,C)$ be observable, where $A\in\mathbb{R}^{n\times n}$. Then there exists a family of observer gains $\{L_\eta\}_{\eta>0}$ such that
\begin{equation*}
\operatorname{spec}(A-L_\eta C)
=
\{-\eta,-2\eta,\ldots,-n\eta\}.
\end{equation*}
Moreover, there exists a polynomial $M(\eta)$ such that
\begin{equation*}
\left\|e^{(A-L_\eta C)t}\right\|
\le
M(\eta)e^{-\eta t},
\qquad \forall t\ge 0,\ \eta\ge 1.
\end{equation*}
In particular, for any fixed $t>0$,
\begin{equation*}
M(\eta)e^{-\eta t}
\longrightarrow 0
\quad \text{as } \eta\to\infty.
\end{equation*}

\begin{proof}
Since $(A,C)$ is observable, for each $\eta>0$ there exists a gain $L_\eta$
such that the closed-loop matrix
\begin{equation*}
A_\eta \coloneqq A-L_\eta C
\end{equation*}
has distinct eigenvalues
\begin{equation*}
\operatorname{spec}(A_\eta)=\{-\eta,-2\eta,\ldots,-n\eta\}.
\end{equation*}
Since these eigenvalues are distinct, $A_\eta$ admits the spectral decomposition
\begin{equation*}
e^{A_\eta t}
=
\sum_{j=1}^{n} e^{-j\eta t}P_j(\eta),
\end{equation*}
where the spectral projectors are given by
\begin{equation*}
P_j(\eta)
=
\prod_{\substack{k=1\\ k\neq j}}^{n}
\frac{A_\eta+k\eta I}{(k-j)\eta}.
\end{equation*}
We now bound the projectors. Since $A_\eta$ is obtained by pole placement
from the polynomial
\begin{equation*}
p_\eta(s)=\prod_{j=1}^{n}(s+j\eta),
\end{equation*}
its entries can be chosen polynomial in $\eta$. Hence each factor
\begin{equation*}
A_\eta+k\eta I
\end{equation*}
has entries polynomial in $\eta$, while the denominator contributes only
powers of $\eta^{-1}$. Therefore, for $\eta\ge 1$, there exists a polynomial
$M_j(\eta)$ such that
\begin{equation*}
\|P_j(\eta)\|\le M_j(\eta).
\end{equation*}
Thus, for some polynomial $M(\eta)$,
\begin{equation*}
\sum_{j=1}^{n}\|P_j(\eta)\|
\le
M(\eta).
\end{equation*}

Using the spectral decomposition and the triangle inequality, we obtain
\begin{equation*}
\|e^{A_\eta t}\|
\le
\sum_{j=1}^{n} e^{-j\eta t}\|P_j(\eta)\|
\le
e^{-\eta t}
\sum_{j=1}^{n}\|P_j(\eta)\|.
\end{equation*}
Therefore,
\begin{equation*}
\|e^{A_\eta t}\|
\le
M(\eta)e^{-\eta t},
\qquad \forall t\ge 0.
\end{equation*}

Finally, since $M(\eta)$ is polynomial, for any fixed $t>0$ we have
\begin{equation*}
M(\eta)e^{-\eta t}\to 0
\quad \text{as } \eta\to\infty.
\end{equation*}
\end{proof}

%%%%%%%%%%%%%%%%%%%%%%%%%%%%%%%%%%%%%%%%%%%%%%%%%%%%%%%%%%%%%%%%%%%%%

\vspace{-0.8cm}

\subsection{Proofs of Lemmas from the Main Text}
The proofs of the lemmas stated in the main text are collected here for completeness.

\noindent\textbf{Lemma \ref{t:decomposition}.} 
(\textit{Orthogonal observability decomposition.})\par
\smallskip

Let \((A,C)\) be a linear pair. Then there exists an orthogonal matrix
\begin{equation*}\label{e:T_orthogonal}
T=\begin{bmatrix}T_o & T_u\end{bmatrix},
\qquad
T^\top T=I,
\end{equation*}
such that
\begin{equation*}\label{e:A_obs_decomp}
T^\top A T
=
\begin{bmatrix}
A_o & 0\\
A_{uo} & A_u
\end{bmatrix},
\end{equation*}
and
\begin{equation*}\label{e:C_obs_decomp}
CT=\begin{bmatrix}C_o & 0\end{bmatrix},
\end{equation*}
where \((A_o,C_o)\) is observable.
\begin{proof}
Let \(\mathcal N\) denote the unobservable subspace of the pair \((A,C)\), i.e.,
\begin{equation*}
\mathcal N=\ker \mathcal O.
\end{equation*}
It is standard that \(\mathcal N\) is \(A\)-invariant, that is,
\begin{equation*}
A\mathcal N\subseteq \mathcal N.
\end{equation*}
Let \(\mathcal V=\mathcal N^\perp\). Then
\begin{equation*}
\mathbb R^n=\mathcal V \oplus \mathcal N,
\end{equation*}
where \(\oplus\) denotes the direct sum. In particular, every \(x \in \mathbb R^n\) admits a unique decomposition
\begin{equation*}
x = x_v + x_n, \qquad x_v \in \mathcal V,\;\; x_n \in \mathcal N.
\end{equation*}

Choose an orthonormal basis of \(\mathcal V\) and an orthonormal basis of \(\mathcal N\), and let
\(T_o\) and \(T_u\) be the matrices whose columns are these basis vectors. Defining
\begin{equation*}
T=\begin{bmatrix}T_o & T_u\end{bmatrix},
\end{equation*}
yields an orthogonal matrix.

Since \(\mathcal N\) is \(A\)-invariant, it follows that \(A T_u \subseteq \mathrm{span}(T_u)\). Consequently, in the basis defined by \(T\), the matrix \(T^\top A T\) has the block lower-triangular form
\begin{equation*}
T^\top A T
=
\begin{bmatrix}
A_o & 0\\
A_{uo} & A_u
\end{bmatrix}.
\end{equation*}

Moreover, since \(\mathcal N \subseteq \ker C\), one has
\begin{equation*}
CT_u=0,
\end{equation*}
and therefore
\begin{equation*}
CT=\begin{bmatrix}C_o & 0\end{bmatrix}.
\end{equation*}

Finally, \((A_o,C_o)\) is observable by construction.
\end{proof}

%%%%%%%%%%%%%%%%%%%%%%%%%%%%%%%%%%%%%%%%%%%%%%%%%%%
\noindent\textbf{Lemma \ref{t:shift_observability}.} 
(\textit{Shift invariance of observability.})\par
\smallskip
Let $A \in \mathbb{R}^{n\times n}$ and $C \in \mathbb{R}^{p\times n}$.
If the pair $(A,C)$ is observable, then for any scalar $\eta \in \mathbb{R}$ the shifted pair
$(A+\eta I,\,C)$ is also observable.

\begin{proof}
Define $A_\eta \coloneqq A+\eta I$. The observability matrices of $(A,C)$ and $(A_\eta,C)$ are
\[
\mathcal{O}
=
\begin{bmatrix}
C\\
CA\\
CA^2\\
\vdots\\
CA^{n-1}
\end{bmatrix},
\qquad
\mathcal{O}_\eta
=
\begin{bmatrix}
C\\
CA_\eta\\
CA_\eta^{2}\\
\vdots\\
CA_\eta^{\,n-1}
\end{bmatrix}.
\]
By the binomial formula,
\[
A_\eta^{\,k}=(A+\eta I)^k=\sum_{j=0}^{k}\binom{k}{j}\eta^{\,k-j}A^{\,j},
\]
and hence
\[
CA_\eta^{\,k}=\sum_{j=0}^{k}\binom{k}{j}\eta^{\,k-j}CA^{\,j}.
\]
Therefore, there exists a block lower-triangular matrix $T\in\mathbb{R}^{np\times np}$ such that
\[
\mathcal{O}_\eta = T\mathcal{O},
\]
whose $(k+1,j+1)$ block entry (each block is $p\times p$) is
\[
T_{k+1,j+1}=
\begin{cases}
\binom{k}{j}\eta^{\,k-j}I_p, & j\le k,\\
0, & j>k,
\end{cases}
\qquad k,j\in\{0,1,\dots,n-1\}.
\]
Since $T$ is lower triangular with diagonal blocks $T_{k+1,k+1}=I_p$, it is nonsingular.
Hence,
\[
\operatorname{rank}(\mathcal{O}_\eta)
=
\operatorname{rank}(T\mathcal{O})
=
\operatorname{rank}(\mathcal{O}).
\]
If $(A,C)$ is observable, then $\operatorname{rank}(\mathcal{O})=n$, and thus $\operatorname{rank}(\mathcal{O}_\eta)=n$.
Therefore, $(A+\eta I,C)$ is observable.
\end{proof}
% \qed 
\smallskip
%%%%%%%%%%%%%%%%%%%%%%%%%%%%%%%%%%%%%%%%%%%%%%%%%%%%%%%%%%%%%

\noindent\textbf{Lemma \ref{t:uniform-euclidean-bound}.} 
(\textit{Uniform exponential Euclidean bound.})\par
\smallskip
Let $\Lambda\in\mathbb{R}^{n\times n}$ be Hurwitz. Then for any
$0<\eta<-\alpha(\Lambda)$, where
\begin{equation*}
\alpha(\Lambda)
\coloneqq
\max_j \Re\big(\lambda_j(\Lambda)\big),
\end{equation*}
there exists a constant $\kappa\ge 1$ such that the solution of
$\dot{\tilde{x}}(t)=\Lambda\tilde{x}(t)$ satisfies
\begin{equation*}
\|\tilde{x}(t)\|
\le
\kappa\,e^{-\eta (t-t_0)}\,\|\tilde{x}(t_0)\|,
\qquad t\ge t_0.
\end{equation*}

\begin{proof}
Fix any $0<\eta<-\alpha(\Lambda)$. Then the shifted matrix
$\Lambda_\eta \coloneqq \Lambda+\eta I$ is Hurwitz.
Choose any $Q=Q^\top\succ0$ and apply \Cref{t:lyap_Lambda} to $\Lambda_\eta$.
There exists a unique $P=P^\top\succ0$ such that
\begin{equation*}
\Lambda_\eta^\top P + P\Lambda_\eta = -Q .
\end{equation*}
Expanding $\Lambda_\eta=\Lambda+\eta I$ gives
\begin{equation*}
\Lambda^\top P + P\Lambda = -Q - 2\eta P \preceq -2\eta P .
\end{equation*}

Define the Lyapunov function $V(t) \coloneqq \tilde{x}(t)^\top P\tilde{x}(t)$.
Along $\dot{\tilde{x}}=\Lambda\tilde{x}$,
\begin{equation*}
\begin{aligned}
\dot V(t)
&=
\tilde{x}(t)^\top\big(\Lambda^\top P + P\Lambda\big)\tilde{x}(t) \\
&\le
-2\eta\,\tilde{x}(t)^\top P\tilde{x}(t) \\
&=
-2\eta\,V(t).
\end{aligned}
\end{equation*}
Hence $V(t)\le e^{-2\eta (t-t_0)}V(t_0)$ for all $t\ge t_0$.
Using $\lambda_{\min}(P)\|\tilde{x}\|^2 \le V \le \lambda_{\max}(P)\|\tilde{x}\|^2$ yields
\begin{equation*}
\|\tilde{x}(t)\|
\le
\sqrt{\frac{\lambda_{\max}(P)}{\lambda_{\min}(P)}}\,
e^{-\eta (t-t_0)}\,\|\tilde{x}(t_0)\|,
\qquad t\ge t_0.
\end{equation*}
Therefore the claim holds with
$\kappa=\sqrt{\mathrm{cond}(P)}$, where
$\mathrm{cond}(P) \coloneqq
\dfrac{\lambda_{\max}(P)}{\lambda_{\min}(P)}$
is the condition number of $P$.
\end{proof}
\smallskip
%%%%%%%%%%%%%%%%%%%%%%%%%%%%%%%%%%%%%%%%%%%%%%%%%%%%%%%%%%%%%%%%%

%%%%%%%%%%%%%%%%%%%%%%%%%%%%%%%%%%%%%%%%%%%%%%%%%%%%%%%%%%%%%%%%

\noindent\textbf{Lemma \ref{t:observer-euclidean-bound}.} 
(\textit{Observer gain with prescribed decay rate.})\par
\smallskip
Let $(A,C)$ be observable. Then for any desired $\eta>0$, there exists a gain
$L$ such that $\Lambda \coloneqq A-LC$ is Hurwitz and satisfies
$\alpha(\Lambda)<-\eta$. Consequently, the estimation-error dynamics
$\dot{\tilde{x}}(t)=\Lambda\tilde{x}(t)$ satisfy
\begin{equation*}
\|\tilde{x}(t)\|
\le
\kappa\,e^{-\eta (t-t_0)}\,\|\tilde{x}(t_0)\|,
\qquad t\ge t_0,
\end{equation*}
for some $\kappa\ge1$.

\begin{proof}
Since $(A,C)$ is observable, the eigenvalues of $A-LC$ can be assigned
arbitrarily in the open left half-plane.
Choose $L$ such that $\alpha(A-LC)<-\eta$, and define $\Lambda \coloneqq A-LC$.
Then $\Lambda$ is Hurwitz with $\alpha(\Lambda)<-\eta$.
The bound follows from \Cref{t:uniform-euclidean-bound}.
\end{proof}
\smallskip

%%%%%%%%%%%%%%%%%%%%%%%%%%%%%%%%%%%%%

\noindent\textbf{Lemma \ref{t:norm_preserve_blocks}.} 
(\textit{Block norm preservation.})\par
\smallskip
Let $T\in\mathbb R^{n\times n}$ be an orthogonal matrix and partition it as
$T=[\,T_o\ \ T_u\,]$, where $T_o\in\mathbb R^{n\times r}$ and
$T_u\in\mathbb R^{n\times (n-r)}$. Then, for any estimation-error
components $\tilde x_o\in\mathbb R^{r}$ and $\tilde x_u\in\mathbb R^{n-r}$,
\begin{equation*}\label{e:Norm_Preserve}
\|T_o \tilde x_o\|=\|\tilde x_o\|,
\qquad
\|T_u \tilde x_u\|=\|\tilde x_u\|.
\end{equation*}

\begin{proof}
Since $T$ is orthogonal,
\begin{equation*}
T^\top T = I_n .
\end{equation*}
Writing $T=[\,T_o\ \ T_u\,]$ gives
\begin{equation*}
T^\top T=
\begin{bmatrix}
T_o^\top T_o & T_o^\top T_u\\
T_u^\top T_o & T_u^\top T_u
\end{bmatrix}
=I_n,
\end{equation*}
hence
\begin{equation*}
T_o^\top T_o=I_r,
\qquad
T_u^\top T_u=I_{n-r}.
\end{equation*}
Therefore,
\begin{equation*}
\|T_o\tilde x_o\|^2
=\tilde x_o^\top T_o^\top T_o \tilde x_o
=\|\tilde x_o\|^2,
\end{equation*}
and similarly,
\begin{equation*}
\|T_u\tilde x_u\|^2
=\tilde x_u^\top T_u^\top T_u \tilde x_u
=\|\tilde x_u\|^2.
\end{equation*}
%%
% Taking square roots yields \eqref{e:Norm_Preserve}.
\end{proof}
%%%%%%%%%%%%%%%%%%%%%%%%%%%%%%%%%%%%%%%%%%%%%%

\noindent\textbf{Lemma \ref{t:O_stack}.} (\textit{Stacked-output observability matrix.})\par
\smallskip
Let  $C_1\in\mathbb R^{p_1\times n}$ and $C_2\in\mathbb R^{p_2\times n}$ and define
$C=\begin{bmatrix}C_1^\top & C_2^\top\end{bmatrix}^\top$.
Consider observability matrices $\mathcal O\coloneqq\mathcal O(A,C)$ and $\mathcal O_i \coloneqq \mathcal O(A,C_i), \quad i=1,2$. 
Then there exists a permutation matrix $M \in \mathbb{R}^{(p_1+p_2)n\times n}$ such that 
\begin{equation*}\label{e:O_stack_identity}
\mathcal O
=M
\begin{bmatrix}
\mathcal O_1\\
\mathcal O_2
\end{bmatrix}.
\end{equation*}
Consequently, if $\operatorname{rank}(\mathcal O)=n$ and $r_i \coloneqq\operatorname{rank}(\mathcal O_i)$, then
\begin{equation*}\label{e:r1_r2_ge_n_from_O}
r_1+r_2\ge n.
\end{equation*}

\begin{proof}
By definition,
\[
\mathcal O(A,C)=
\begin{bmatrix}
C\\
CA\\
\vdots\\
CA^{n-1}
\end{bmatrix}
=
\begin{bmatrix}
C_1\\
C_2\\
C_1A\\
C_2A\\
\vdots\\
C_1A^{n-1}\\
C_2A^{n-1}
\end{bmatrix}
=
M\begin{bmatrix}
\mathcal O(A,C_1)\\
\mathcal O(A,C_2)
\end{bmatrix},
\]
which proves \eqref{e:O_stack_identity}. Since multiplication by permutation matrices does not change rank, if $\operatorname{rank}(\mathcal O)=n$, then
\[
n=\operatorname{rank}\!\begin{bmatrix}\mathcal O_1\\ \mathcal O_2\end{bmatrix}
\le \operatorname{rank}(\mathcal O_1)+\operatorname{rank}(\mathcal O_2)=r_1+r_2.
\]
\end{proof}

%%%%%%%%%%%%%%%%%%%%%%%%%%%%%%%%%%%%%%%%%%%%%%%%%%%%%

%%%%%%%%%%%%%%%%%%%%%%%%%%%%%%%%%%%%%%%%%%%%%%%%%%

\noindent\textbf{Lemma \ref{t:feasibility_eta}.} (\textit{Feasibility of the design inequalities.})\par
\smallskip
For each cycle $k$, the inequalities
\eqref{e:eta1_bound} and \eqref{e:eta2_bound} are feasible.

\begin{proof}
We show feasibility of \eqref{e:eta1_bound}; the proof for
\eqref{e:eta2_bound} is analogous.

By Lem. A.2, the prefactor $\kappa_k^{(1)}$
can be chosen as
\[
\kappa_k^{(1)} = M(\eta_k^{(1)}),
\]
where $M$ is a polynomial. Hence, there exist constants $C>0$, $d>0$,
and $\bar{\eta}>0$ such that
\[
M(\eta)\le C\eta^d,
\qquad \forall \eta\ge \bar{\eta}.
\]
Therefore, for all $\eta_k^{(1)}\ge \bar{\eta}$,
\[
\ln\!\left(\frac{\kappa_k^{(1)}}{\delta_1}\right)
\le
\ln\!\left(\frac{C}{\delta_1}\right)+d\ln \eta_k^{(1)}.
\]
Substituting into \eqref{e:eta1_bound}, it is sufficient to require
\[
\eta_k^{(1)}
\ge
\frac{\mu^{(1)}(1-s_k)}{s_k}
+
\frac{1}{s_k\tau_k}\ln\!\left(\frac{C}{\delta_1}\right)
+
\frac{d}{s_k\tau_k}\ln \eta_k^{(1)}.
\]
For each fixed $k$, this is of the form in Lem. A.1 with
\[
c_1=
\frac{\mu^{(1)}(1-s_k)}{s_k}
+
\frac{1}{s_k\tau_k}\ln\!\left(\frac{C}{\delta_1}\right),
\qquad
c_2=\frac{d}{s_k\tau_k}.
\]
Hence, by Lem. A.1, there exists $\eta_k^{(1)}>0$
satisfying the above inequality. Choosing additionally
$\eta_k^{(1)}\ge \bar{\eta}$, we conclude that
\eqref{e:eta1_bound} is feasible.
\end{proof}
\smallskip

%%%%%%%%%%%%%%%%%%%%%%%%%%%%%%%%%%%%%%%%%%%%%%%

\noindent\textbf{Lemma \ref{t:observer-sdp-design}.} 
(\textit{LMI synthesis of observer gain.})\par
\smallskip
Let $(A,C)$ be observable and fix $\eta>0$.
Consider the semidefinite program
\begin{align*}
\min_{P,Y,\theta} \quad & \theta \\
\text{subject to}\quad
& P=P^\top\succ0, \\
& A^\top P+PA-C^\top Y^\top-YC+2\eta P \preceq 0, \\
& I \preceq P \preceq \theta I.
\end{align*}
Let $(P^\star,Y^\star,\theta^\star)$ be an optimal solution and define
\begin{equation*}
L^\star \coloneqq (P^\star)^{-1}Y^\star.
\end{equation*}
Then the estimation error satisfies
\begin{equation*}
\|\tilde{x}(t)\|
\le
\kappa^\star\,e^{-\eta (t-t_0)}\|\tilde{x}(t_0)\|,
\qquad t\ge t_0,
\end{equation*}
with $\kappa^\star=\sqrt{\theta^\star}$.

\begin{proof}
From $I\preceq P\preceq \theta I$ we have $\lambda_{\min}(P)\ge 1$ and
$\lambda_{\max}(P)\le \theta$, hence $\mathrm{cond}(P)\le \theta$ and
$\sqrt{\mathrm{cond}(P)}\le \sqrt{\theta}$.
With $L^\star=(P^\star)^{-1}Y^\star$, the LMI
$A^\top P+PA-C^\top Y^\top-YC+2\eta P\preceq 0$ becomes
\[
(A-L^\star C)^\top P^\star+P^\star(A-L^\star C)\preceq -2\eta P^\star.
\]
Applying the Lyapunov argument in \Cref{t:uniform-euclidean-bound}
yields
$\|\tilde{x}(t)\|\le \sqrt{\mathrm{cond}(P^\star)}e^{-\eta(t-t_0)}\|\tilde{x}(t_0)\|$,
and since $\sqrt{\mathrm{cond}(P^\star)}\le\sqrt{\theta^\star}$ the claim follows
with $\kappa^\star=\sqrt{\theta^\star}$.
\end{proof}

%%%%%%%%%%%%%%%%%%%%%%%%%%%%%%%%%%%%%%%%%%%%%%%%%%%%%%%%%%%%%%%%%
% \subsection{Auxiliary lemmas}
% \begin{lemma}[Polynomial decay bound]\label{t:polynomial-prefactor}
% Let $(A,C)$ be observable, where $A \in \mathbb{R}^{n\times n}$, and let $\tau>0$ be fixed. For each $\eta>0$, choose an observer gain $L_\eta$ such that the matrix $A - L_\eta C$ has characteristic polynomial $(s+\eta)^n$. Then there exists a polynomial $M(\eta)$ such that
% \begin{equation}
% \left\|e^{(A - L_\eta C)\tau}\right\|
% \le
% M(\eta)\,e^{-\eta \tau}, \qquad \forall \eta>0.
% \end{equation}
% In particular,
% \begin{equation}
% M(\eta)\,e^{-\eta \tau} \;\longrightarrow\; 0
% \quad \text{as } \eta \to \infty.
% \end{equation}
% \end{lemma}
% \smallskip

%%%%%%%%%%%%%%%%%%%%%%%%%%%%%%%%%%%%%%%%%%%%%%%%%%%%%%%%%%%%%%%%%%%%%%%

% \begin{lemma}\label{t:eta-log}
% Let $c_1>0$ and $c_2>0$. Then the inequality
% \begin{equation}
% \eta \;\ge\; c_1 + c_2 \ln \eta
% \end{equation}
% admits a solution.
% \end{lemma}

% \begin{proof}
% Define $f(\eta):=\eta - c_2\ln\eta$. Then
% \[
% f'(\eta)=1-\frac{c_2}{\eta},
% \]
% so $f'(\eta)>0$ for all $\eta>c_2$, and hence $f$ is strictly increasing on $(c_2,\infty)$. Moreover, $f(\eta)\to\infty$ as $\eta\to\infty$. Therefore, there exists $\eta>c_2$ such that $f(\eta)\ge c_1$, which proves the claim.
% \end{proof}
% \smallskip

%%%%%%%%%%%%%%%%%%%%%%%%%%%%%%%%%%%%%%%%%%%%%%%%%%%%%%%%%%%%%%%%%%%%%%%%%%%%%%%%%%%%%%%%%%%%%%%%%%%%%%%%%%%%%%%%%%%%%%%%%%%%%%%%%%%%%%%%%%%%%%%%%%%%%%%%
\newpage
\bibliographystyle{IEEEtran}
% \bibliography{scc,sccmaster,scctemp,Myref}
\bibliography{Bib/MyRef}
% \bibliography{MyRef.bib}

\end{document}